\documentclass[10pt]{article}
\usepackage[top=0.9in,bottom=1.1in,left=1.05in,right=1.05in]{geometry}
\usepackage{amsmath,amssymb}
\usepackage{amsthm}
\usepackage{threeparttable}
\usepackage{latexsym}
\usepackage{mathrsfs,dsfont}
\usepackage{cancel}
\usepackage{comment}
\usepackage{float}
\usepackage{graphicx}
\usepackage{epstopdf}
\usepackage{color}
\usepackage{enumerate}
\usepackage{natbib}
\usepackage{mathtools}
\mathtoolsset{showonlyrefs}
\DeclareGraphicsExtensions{.eps}
\numberwithin{equation}{section}
\DeclareMathOperator*{\esssup}{ess\,sup}
\newcommand{\MCG}{\mathcal{G}}
\newcommand{\MCC}{\mathcal{C}}
\newcommand{\MCA}{\mathcal{A}}

\newcommand{\MCF}{\mathcal{F}}
\newcommand{\MCO}{\mathcal{O}}
\newcommand{\MCD}{\mathcal{D}}

\newcommand{\MCK}{\mathcal{K}}
\newcommand{\EE}{\mathbb{E}}
\newcommand{\PP}{\mathbb{P}}
\newcommand{\RR}{\mathbb{R}}
\newcommand{\NN}{\mathbb{N}}

\newcommand{\Ltx}{\mathcal{L}_{t,x}}

\newcommand{\ltwonorm}[1]{\left\lVert#1\right\rVert_2}

\newcommand{\pnorm}[1]{\left\lVert#1\right\rVert_p}
\newcommand{\pqnorm}[1]{\left\lVert#1\right\rVert_{pq}}
\newcommand{\prnorm}[1]{\left\lVert#1\right\rVert_{pr}}

\newcommand{\Vy}{V^{\eps}}

\newcommand{\vz}{v^{(0)}}
\newcommand{\vo}{v^{(1)}}

\newcommand{\pz}{{\pi^{(0)}}}

\newcommand{\Vyl}{V^{\pz,\eps}}

\newcommand{\eps}{\epsilon}

\newcommand{\abs}[1]{\left|#1\right|}
\newcommand{\average}[1]{\left\langle#1\right\rangle}

\newcommand{\mc}[1]{\mathcal{#1}}

\newcommand{\ud}{\,\mathrm{d}}

\newcommand{\Wh}[1]{W^{(H)}_{#1}}

\newcommand{\half}{\frac{1}{2}}
\newcommand{\Yh}[1]{Y^{\eps,H}_{#1}}
\newcommand{\Yht}[1]{\widetilde Y^{\eps,H}_{#1}}
\newcommand{\kereps}{\MCK^\eps}

\newtheorem{theo}{Theorem}[section]
\newtheorem{lem}[theo]{Lemma}

\newtheorem{rem}[theo]{Remark}
\newtheorem{prop}[theo]{Proposition}

\begin{document}

\title{\vspace{-50pt} Portfolio Optimization under Fast Mean-reverting and Rough Fractional Stochastic
	Environment}
%\author{Ruimeng Hu\\Department of Statistics and Applied Probability \\ University of California, Santa Barbara}
\author{Jean-Pierre Fouque\thanks{Department of Statistics \& Applied Probability,
 University of California,
        Santa Barbara, CA 93106-3110, {\em fouque@pstat.ucsb.edu}. Work  supported by NSF grant DMS-1409434 and DMS-1814091.}
        \and Ruimeng Hu\thanks{Department of Statistics, Columbia University, New York, NY 10027-4690, {\em rh2937@columbia.edu}. The work was mainly done when RH was a graduate student at the University of California, Santa Barbara.}
        }
\date{\today}
\maketitle

\begin{abstract}
	
	Fractional stochastic volatility models have been widely used to capture the non-Markovian  structure revealed from financial time series of realized volatility. On the other hand, empirical studies have identified scales in stock price volatility: both fast-time scale on the order of days and slow-scale on the order of months. So, it is natural to study the portfolio optimization problem under the effects of dependence behavior which we will model by fractional Brownian motions with Hurst index $H$, and in the fast or slow regimes characterized by small parameters $\eps$ or $\delta$. For the slowly varying volatility with $H \in (0,1)$, it was shown that the first order correction to the problem value contains two terms of order $\delta^H$, one random component and one deterministic function of state processes, while for the fast varying case with $H > \half$, the same form holds at order $\eps^{1-H}$. This paper is dedicated to the remaining case of a fast-varying  rough environment ($H < \half$) which exhibits  a different behavior.  We show that, in the expansion, only one deterministic term of order $\sqrt{\eps}$ appears in the first order correction. 
%The three piece of work together form a whole picture of analyzing portfolio optimization in both fast and slow fractional environments.

\end{abstract}


\section{Introduction}\label{sec_intro}
Portfolio optimization  in  continuous time  was originally studied by \cite{Me:69, Me:71}. In this celebrated work, explicit solutions are provided on how to allocate wealth between risky and risk-less assets, and how to consume wealth so that expected utility is maximized. The model there for the underlying assets is the Black--Scholes--Merton model, and the utility functions are of specific type, for instance, of Constant Relative Risk Aversion (CRRA) type.

In this paper, we study the Merton problem with power utility under a non-Markovian stochastic environment that is fast mean-reverting. Such proposed modeling is supported by recent studies. It is shown that stochastic volatility has rapidly decaying correlations at the origin. A common approach for modeling short-range dependence is by using fractional Brownian motion (fBm) processes with $H < 1/2$, see \cite{roughvol}. On the other hand, empirical study in \cite{FoPaSiSo:03} shows that, in order to fit well the implied volatility, it is appropriate to consider at least short-time scales volatility. To combine these two features, we let the underlying asset follow a Geometric Brownian motion-like model, and we let its return and volatility be driven by a fast-varying factor  $\Yh{t}$ characterized as the solution of
\begin{equation*}
\ud \Yh{t} = -\frac{a}{\eps}\,\Yh{t} \ud t + \frac{1}{\eps^{H}}\ud \Wh{t}.
\end{equation*}
For a full justification/discussion about this model, we refer to \cite{GaSo:17}, where $\Yh{t}$ is firstly introduced. Here $\eps \ll 1$ is a small parameter, which makes $\Yh{t}$ fast-varying and its mean-reversion time scale proportional to $\eps$, and, $\Wh{t}$ denotes a fBm with Hurst index $H \in (0, \half)$. The solution to the above stochastic differential equation (SDE) is actually a fractional Ornstein--Uhlenbeck process. Some properties regarding $\Yh{t}$ are discussed in Section~\ref{sec_fastfOU} as a preparation for the asymptotic derivations in Section~\ref{sec_asymppower}. For further references on fBm and fOU processes, we refer to \cite{MaVa:68, ChKaMa:03,Co:07,BiHuOkZh:08,KaSa:11}. 

\medskip
\noindent{\bf Motivation and related literature.} Main features and reasons to consider such a problem setup are the following.

Firstly, the fOU process $\Yh{t}$ is Gaussian and it admits a convenient moving-average representation with respect to a well-studied kernel (cf. \eqref{eq_Yh}). This maintains  analytic tractability and simplifies the derivation of needed estimates. However, the modeling of $\Yh{t}$ is not limited to fOU processes, and more general kernels can be considered. For example, an alternative model is analyzed in \cite[Appendix B]{GaSo:17} for the linear pricing problem.

Secondly, we choose to focus on power utilities and one factor model, which means the return and volatility of the underlying asset are driven by only one process. The reason is mainly due to tractability.  In this case, a convenient martingale distortion transformation (MDT) is available to express the problem value and the optimal strategy, even under the non-Markovian modeling of $\Yh{t}$. However, we remark that partial results can still be obtained under general utilities where MDT is not available, following a similar argument as in our previous work \cite{FoHu:17, FoHu2:17}. The details of this generalization will not be included here.

Thirdly, the fOU process $\Yh{t}$ is rough ($H < \half$) and fast mean-reverting ($\eps$ small). This is the missing case in our previous work \cite{FoHu:17,FoHu2:17}, where the asset allocation problem is studied under a slowly varying fractional stochastic environment (fSE) with $H \in (0, 1)$, corresponding to $\eps:=1/\delta$ large, and under a fast varying fSE ($\eps$ small) with $H>\half$. Therefore, this paper completes the full picture 
of the analysis of the portfolio optimization problem in single-factored fractional stochastic environments.

Fourthly, although it is natural to consider multiscale factor models for risky assets, with a fast factor and a slow factor as in \cite{FoSiZa:13} in a Markovian framework, the analysis requires more technical details, as the MDT is not available. This will be presented in another paper in preparation (\cite{Hu:XX}). 

\medskip
We now summarize the related existing literature in the following table. 
\begin{table}[H]
	\centering
	\caption{Problems, models and expansion results. The acronyms used are: SV = Stochastic Volatility, SE = Stochastic Environment, AO = Asymptotic Optimality.}\label{table1}
	\begin{threeparttable}
	\begin{tabular}{|l|c|c|c|}\hline
		Paper  & Problem & Model & Form of Solution\tnote{$\ast$} \\ \hline\hline 
		\cite{FoPaSiSo:11} & {\bf Linear} & SV + Multiscale & $P^{(0)} + \sqrt{\eps} P^{(1,0)} + \sqrt{\delta} P^{(0,1)} $  \\ 
		\cite{GaSo:15} & Option &  fSV + Slow ($ H \in (0,1)$) & $P^{(0)} + \phi^\delta + \delta^H P^{(1)}$\tnote{}\\ 
		\cite{GaSo:16} & Pricing & fSV + Fast ($H > \half$) & $P^{(0)} + \phi^\eps + \eps^{1-H} P^{(1)}$  \\ 
		\cite{GaSo:17} & & fSV + Fast ($H < \half$) & $P^{(0)} + \sqrt{\eps} P^{(1)}$ \\ \hline
		\cite{FoSiZa:13} & & SE + Multiscale & $\vz + \sqrt{\eps} v^{(1,0)} + \sqrt{\delta} v^{(0,1)}$\tnote{$\dagger$}\\ 
		\cite{FoHu:16} &{\bf Nonlinear} & SE + Slow & AO of a zeroth order strategy \\
		\cite{Hu:18} & Portfolio & SE + Fast & AO of a zeroth order strategy \\ 
		\cite{FoHu:17} &Optimization & fSE + Slow ($H \in (0,1)$) & $\vz + \phi^\delta + \delta^H \vo$\\ 
		\cite{FoHu2:17} & & fSE + Fast ($ H > \half$) & $\vz + \phi^\eps + \eps^{1-H} \vo$ \\ 
		This paper & & fSE + Fast ($H < \half$) &  $\vz +\sqrt{\eps} \vo$\\ \hline		 		
	\end{tabular}
	  \begin{tablenotes}
	  	\item[$\ast$] We denote by $h^{(0)}$ the leading order term in the expansion, and by $h^{(0,1)}$, $h^{(1,0)}$, $h^{(1)}$ the first order corrections,  where $ h = P $ for option pricing, and $h = v$ for problem value of the optimization problem. The notation $\phi^\delta$ (\emph{resp.} $\phi^\eps$) means a random component of order $\delta^H$ (\emph{resp.} $\eps^{1-H}$).
	  	\item[$\dagger$] The expansion is heuristic except for the case of power utility and one factor.
	  \end{tablenotes}
	  \end{threeparttable}
\end{table}

\medskip
\noindent{\bf Main results.} In this paper, we focus on one-factor models and we study the effect of a fast time-scale on the optimal allocation problem under power utility. When the problem is Markovian, a PDE approach is preferred, since after a distortion transformation, firstly discovered in \cite{Za:99}, the PDE becomes linear where perturbation techniques usually work well. However, in the current setting, the fast factor is driven by a rough fBm and possesses short-range dependence.  Nevertheless, the nice MDT is available and gives a representation of the value process as well as the optimal strategy. This was proved by \cite{Te:04} via a conditional H\"{o}lder inequality, and by \cite{FrSc:08} via a BSDE approach in the case of exponential utility. Recently, it has been restated in \cite{FoHu:17} under the  setup \eqref{def_St} with a short proof based on a  verification argument. 

Starting by applying MDT, we obtain the representation of the problem value and the optimal portfolio. We then expand them using the ``ergodic property'' of $\Yh{t}$, and we deduce approximation results for both quantities. Unlike in the long-range dependent case $H > \half$ in \cite{FoHu2:17}, here, the first order correction to the value process appear at $\sqrt{\eps}$ and contains only one term, which is a explicit function of the state processes. And, surprisingly, there is no correction term at order $\sqrt{\eps}$ for the optimal strategy. However, we are still able to show that the leading order strategy $\pi_t^{(0)}$ itself generates the value process up to its first order $\sqrt{\eps}$ correction, which is obtained by the ``epsilon-martingale decomposition'' method. This approach was firstly introduced in \cite{FoPaSi:00,FoPaSi:01} and frequently used in problems with small parameters, especially in  non-Markovian settings as in \cite{GaSo:15,GaSo:16,GaSo:17}. 

We remark that our expansion is only valid when $H \in (0, \half)$ where $\Yh{t}$ does not appear in the leading order nor in the correction of the value process. Moreover, up to order $\sqrt{\eps}$, the expansion does not really depend on $H$, except through constants $\overline D$ and $\overline \lambda$, see \eqref{thm_Vtpowerexpansion} -\eqref{eq_Vtpower}. We also observe that the limit  $H \uparrow \half$ does not commute with the limit $\eps \to 0$ (see Section \ref{sec_comparisonMarkov}), which makes  impossible to recover the results in the Markovian case provided in \cite{FoSiZa:13}.

\medskip
\noindent{\bf Organization of the paper.} In Section~\ref{sec_SVmerton}, we revisit the martingale distortion transformation under general stochastic volatility models. This is derived in the Markovian case in \cite{Za:99}, and in non-Markovian settings in \cite{Te:04,FrSc:08,FoHu:17}. Then the fast mean-reverting rough fractional stochastic environment (RFSE) is precisely described, that is, $\Yh{t}$ follows a $\eps$-scaled fractional Ornstein--Uhlenbeck process. 
In Sections~\ref{sec_asympVy} and \ref{sec_asymppi}, we derive the asymptotic results under this modeling for the value process and optimal portfolio respectively. The asymptotic optimality of the leading order strategy $\pz$ is proved within all admissible strategies in Section \ref{sec_asymppzpower}. Finally, we compare the results with the Markovian case $H = \half$, and we comment on the influence of rough fractional model. We make conclusive remarks in Section~\ref{sec_conclusion}.

\section{Problem setup and preliminaries}\label{sec_SVmerton}

Denote by $S_t$ the risky asset price at time $t$ with both returns and volatility driven by a stochastic factor $Y_t$:
 \begin{align}
 \ud S_t = S_t\left[\mu(Y_t)\ud t  + \sigma(Y_t)\ud W_t\right]. \label{def_St}
 \end{align}
Here $Y_t$ is a general stochastic process adapted to the filtration $\MCG_t \equiv \sigma(W_s^Y, s\leq t)$ of a Brownian motion $W^Y$.  The two Brownian motions $W$ and $W^Y$ are imperfectly correlated:
\begin{equation}\label{def_BMcorrelation}
\ud \average{W_t, W_t^Y} = \rho \ud t, \quad \abs{\rho}<1.
\end{equation}
We define $(\MCF_t)$ as the natural filtration generated by the two Brownian motions $(W_t, W_t^Y)$, and we shall use $Y_t$ to model the one factor stochastic environment. Later, in Section~\ref{sec_asymppower}, $Y_t$ will be replaced by $\Yh{t}$, which will be a fast mean-reverting   fOU process. 

To formulate the Merton problem under such a stochastic environment, we introduce the following notations. Denote by $\pi_t \in \MCF_t$ the amount of money invested in the risky asset $S_t$ at time $t$, and by $X_t^\pi$ the corresponding wealth process. The rest, $X_t^\pi - \pi_t$ is put into the bank account earning a constant interest rate $r$. Then, under self-financing, and, without loss of generality, assuming a zero interest rate, the dynamics of $X_t^\pi$ is given by:
\begin{equation}
\ud X_t^\pi = \pi_t\mu(Y_t) \ud t + \pi_t\sigma(Y_t) \ud W_t.\label{def_XtunderY}
\end{equation}
In the Merton problem, the agent aims at finding the optimal allocation $\pi$ so as to optimize the expected utility of her terminal wealth $X_T^\pi$. Mathematically, it consists in identifying the value process $V_t$ defined by
\begin{equation}\label{def_Vt}
V_t \equiv \esssup_{\pi \in \MCA_t}\EE\left[U(X_T^\pi)\vert \MCF_t\right],
\end{equation}
and the corresponding optimal strategy $\pi^\ast$. Here $U(\cdot)$ is a utility function describing the agent's preference. Throughout this paper, we shall work with power utilities, namely 
\begin{equation}\label{def_power}
U(x) = \frac{x^{1-\gamma}}{1-\gamma}, \quad \gamma >0, \quad \gamma \neq 1.
\end{equation}
The set $\MCA_t$ is the collections of admissible strategies: 
\begin{equation}\label{def_MCA}
\MCA_t \equiv \left\{\pi \text{ is } (\MCF_t)\text{-adapted} : X_s^\pi  \text{ stays nonnegative } \forall s \geq t, \text{ given } \MCF_t\right\},
\end{equation}
with the following integrability condition:
	\begin{align}\label{assump_strategies}
	\sup_{t\in[0,T]}\EE\left[ \left(X_t^\pi\right)^{2p(1-\gamma)}\right] < +\infty,\; for \; some  \;p > 1, \quad and \quad \EE \left[\int_0^T\left(X_t^\pi\right)^{-2\gamma}    \pi_t^2  \sigma^2(Y_t)\ud t \right] < \infty.
	\end{align}
In addition, Assumption 2.1 in \cite{FoHu2:17} is enforced throughout the paper.

In \cite{FoHu:17}, the representation of value process \eqref{def_Vt} and optimal strategy $\pi^\ast$ are given via a \emph{martingale distortion transformation}. The rest of this section is  a preparation for the main results presented in the next section. The martingale distortion transformation is reviewed, and the fast mean-reverting RFSE is introduced. 

%as well as the fractional Brownian motion (fBm) and fractional Ornstein-Uhlenbeck (fOU) processes, which are building blocks of the fast mean-reverting and rough fractional stochastic environment.
 
%and justify that the assumptions required for the martingale distortion transformation are satisfied by this specific fractional stochastic environment.

\subsection{Martingale distortion transformation}\label{sec_martdistort}
The martingale distortion transformation was derived in \cite{Te:04} with a slightly different utility function, and, recently, stated in \cite{FoHu:17} under the same setup as in this paper. Since the results in next section heavily rely on this transformation, for reader's convenience, we re-state it here. 

Let $\widetilde \PP$ be an equivalent probability measure determined by the Radon--Nikodym derivative
\begin{equation}\label{def_Ptilde}
	\frac{\ud \widetilde \PP}{\ud \PP} = \exp\left\{-\int_0^T a_s \ud W_s^Y - \half \int_0^T a_s^2 \ud s \right\},
\end{equation}
with $a_t = -\rho\left(\frac{1-\gamma}{\gamma}\right)\lambda(Y_t)$, and Sharpe-ratio $\lambda(\cdot) = \mu(\cdot)/\sigma(\cdot)$. 
%being
%	\begin{equation}\label{def_at}
%		
%	\end{equation}
		%is bounded and $\MCG_t$-adapted.
		%Therefore, $\widetilde W_t^Y :=  W_t^Y+ \int_0^t a_s \ud s$ is a $\widetilde\PP$-Brownian motion. 
Then, define the $\widetilde \PP$-martingale
		\begin{equation}\label{def_Mtmartdistort}
		M_t \equiv \widetilde \EE\left[\left.e^{\frac{1-\gamma}{2q\gamma}\int_0^T \lambda^2(Y_s)\ud s} \right\vert \MCG_t\right], 
		\end{equation}
		where the parameter $q$ is given in terms of the utility's relative risk-aversion $\gamma$ and the correlation coefficient $\rho$ by 
		\begin{equation}\label{def_q}
					q=\frac{\gamma}{\gamma+(1-\gamma)\rho^2}.
		\end{equation}
		The martingale $(M_t)_{t\in[0,T]}$ admits the representation
		\begin{equation}\label{def_xi}
		\ud M_t = M_t\xi_t \ud \widetilde W_t^Y.
		\end{equation}

\begin{prop}[Martingale Distortion Transformation]\label{prop_martdistort}
	Under model assumptions, the value process $V_t$ defined in \eqref{def_Vt} is expressed by
	\begin{equation}\label{def:Vcorrelated}
	V_t=\frac{X_t^{1-\gamma}}{1-\gamma} \left[\widetilde{\EE}\left(\left.e^{\frac{1-\gamma}{2q\gamma}\int_t^T\lambda^2(Y_s)\ud s}\right\vert \mc{G}_t\right)\right]^q.
	\end{equation}
	The expectation
	$\widetilde \EE[\cdot]$ is computed with respect to $\widetilde{\PP}$ introduced in \eqref{def_Ptilde}, and the parameter $q$ is given in \eqref{def_q}. 
	
	The optimal strategy $\pi^\ast$ is
	\begin{equation}\label{def_pioptimal}
	\pi^\ast_t = \left[\frac{\lambda(Y_t)}{\gamma \sigma(Y_t)} + \frac{\rho q \xi_t}{\gamma \sigma(Y_t)}\right] X_t,
	\end{equation}
	where $\xi_t$ is given by the Martingale Representation Theorem in \eqref{def_xi}.
\end{prop}
\begin{proof}
	See \cite[Proposition~2.2]{FoHu:17} for a detailed proof. Discussions about special cases (uncorrelated $\rho=0$, degenerate Sharpe-ratio $\lambda(y) = \lambda_0$, etc.) and generalization to multi-asset case can also be found therein. 
\end{proof}

\begin{rem}
The model assumptions consist of the existence of a unique strong solution to \eqref{def_St}, the coincidence between the filtration generated by $Y_t$ and $\MCG_t$, the regularity condition on $\lambda(\cdot)$ and integrability condition on $\xi_t$:
\begin{equation}\label{assump_xi}
\EE[e^{c\int_0^T \xi_t^2 \ud t}] < \infty, \text{ for some constant }c.
\end{equation}
They are actually the same as in our previous paper \cite{FoHu2:17}. So, for the sake of space, we omit the detailed description here, and refer to Assumption~2.1 and Remark~2.2 therein for further discussion.
\end{rem}

%
%\begin{assump}\label{assump_power}\quad
%	\begin{enumerate}[(i)]
%		\item\label{assump_St} The SDE \eqref{def_St} for $S_t$ has a unique strong solution, in other words, $$S_t=S_0e^{\int_0^t(\mu(Y_s)-\half\sigma^2(Y_s))\ud s+\int_0^t\sigma(Y_s)\ud W_s}$$
%		exists for all $t\in [0,T]$. 
%		\item\label{assump_filtration} Assume the filtration generated by $(Y_s)_{s \leq t}$ is also $\MCG_t$, and the volatility function $\sigma(\cdot)$ is injective.
%		
%		\item\label{assump_lambda} The Sharpe ratio $\lambda(\cdot) := \mu(\cdot)/\sigma(\cdot)$ is assumed to be bounded and $C^2(\RR)$. Also, the derivatives $\lambda'$ and $\lambda''$ are assumed bounded. 
%
%		We assume
%		\begin{equation*}
%		\EE\left[e^{c_\xi\int_0^T \xi_t^2\ud t}\right] < \infty, 
%		\end{equation*}
%		where the constant $c_\xi$ is given by 
%		$c_\xi = \frac{16 (1-\gamma)^2\rho^2 p^2q^2}{\gamma^2}$ for  $\gamma < 1$, and $c_\xi = \frac{16 (1-\gamma)^2\rho^2 p^2q^2}{\gamma^2} - \frac{4p(1-\gamma)}{\gamma^2}$ for  $\gamma >1$. The parameter $p$ is introduced in \eqref{assump_strategies} and $q$ is defined in terms of $\gamma$ and $\rho$ by
%		
%		Note that $q$ is the usual ``distortion'' exponent firstly introduced in \cite{Za:99}.	
%	\end{enumerate}
%\end{assump}
%
%

\subsection{Modeling of the fast mean-reverting RFSE}\label{sec_fastfOU}

To accommodate the fast mean-reverting property in the stochastic environment, we introduce the small parameter $\eps$ in the modeling of $Y_t$, and we switch to the notation  $\Yh{t}$ to emphasize this dependence. The other superscript $H$ comes from the Hurst index of the fractional Brownian motion that drives $\Yh{t}$, see \eqref{def_YhSDE}. Then, in the rest of this paper, $\Yh{t}$ plays the role of $Y_t$ in equation \eqref{def_St}. 

Following  \cite{GaSo:17}, we define our fast factor $\Yh{t}$ by 
\begin{equation}
\Yh{t}= \int_{-\infty}^t \kereps(t-s) \ud W_s^Y, \quad \kereps(t) = \frac{1}{\sqrt{\eps}}\mc{K}\left(\frac{t}{\eps}\right),\label{eq_Yh}
\end{equation}
where $\eps \ll 1$ is a small parameter,  $\MCK(t)$ is a non-negative kernel taking the form
\begin{equation}\label{def_kernel}
\mc{K}(t) = \frac{1}{\Gamma(H+\half)} \left[t^{H-\half} - a \int_0^t (t-s)^{H-\half}e^{-as}\ud s\right],
\end{equation}
with $H \in(0, \half)$, $(W_t^Y)_{t \in \RR^+}$ is the standard Brownian motion (Bm) that is correlated with $W_t$ as given in \eqref{def_BMcorrelation}, and $(W_t^Y)_{t \in \RR^-} := (B_{-t})_{t \in \RR^-}$ is another Bm independent of $(W_t^Y)_{t \in \RR^+}$ and $(W_t)$. The case $H = \half$ recovers the usual Markovian OU process.
 %This also clears why we put a superscript $H$  to the notation $\Yh{t}$. 
Such $\Yh{t}$ is actually the unique (in distribution) stationary solution to the rescaled fractional Ornstein--Uhlenbeck (fOU) SDE:
\begin{equation}\label{def_YhSDE}
\ud \Yh{t} = -\frac{a}{\eps}\,\Yh{t} \ud t + \frac{1}{\eps^{H}}\ud \Wh{t},
\end{equation}
where $W_t^{(H)}$ is the fractional Brownian motion (fBm) with Hurst index $H$. Properties regarding this process (with or without scaling) has been widely studied, for instance, see \cite{ChKaMa:03,GaSo:17}. For the sake of simplicity, here, we only mention results that are related to our derivations in the sequel. 

The process $\Yh{t}$ is Gaussian with mean zero and (co)variance structure:
\begin{align}
& \EE\left[\left(\Yh{t}\right)^2\right]  \equiv \sigma_{ou}^2 = \half a^{-2H}\sin(\pi H)^{-1}, \quad \EE\left[\Yh{t}\Yh{t+s}\right] =\sigma^2_{ou}\MCC_Y\left(\frac{s}{\eps}\right), \\
& \MCC_Y(s) \equiv \frac{2\sin(\pi H)}{\pi} \int_0^\infty \cos(asx)\frac{x^{1-2H}}{1+x^2}\ud x.\label{eq_Zhcovar}
\end{align}
The form of covariance shows that the natural scale of $\Yh{t}$ is $\eps$ as desired. The variance of $\Yh{t}$ stays invariant, which validates the way $\MCK^\eps(t)$ is rescaled in \eqref{eq_Yh}.
The kernel $\MCK(t)$ is in $L^2$ with $\int_0^\infty \MCK^2(u) \ud u = \sigma_{ou}^2$, and  $\MCK(t) \in L^1$, when $H < \half$. 

Denote by $\overline\lambda^2$ the average with respect to the invariant distribution of fOU process $\mc{N}(0, \sigma_{ou}^2)$:
\begin{equation}\label{def_averagelambda}
\overline{\lambda}^2 \equiv \int_\RR \lambda^2(z) \frac{1}{\sqrt{2\pi}\sigma_{ou}}e^{-\frac{z^2}{2\sigma^2_{ou}}} \ud z
\end{equation}
Then, the following differences between time averages and  spacial averages are of importance in the derivations:
\begin{align}
& I_t^\eps \equiv \int_0^t \left( \lambda^2(\Yh{s}) - \overline{\lambda}^2\right) \ud s, \label{def_i}\\
& \phi_t^\eps \equiv \EE\left[\int_t^T \lambda^2(\Yh{s}) - \overline{\lambda}^2 \ud s \bigg\vert \MCG_t\right]. \label{def_phi}
\end{align}
By the ergodicity of $\Yh{t}$, these differences are small and of order $\eps^{1-H}$, see Lemma~\ref{lem_moments}. More properties and estimates regarding $\Yh{t}$ are also stated therein.

Furthermore, our fast factor $\Yh{t}$ satisfies the model assumption, and we state it as follows.

\begin{lem}\label{lem_xi}
	Under model assumptions, the fast mean-reverting stationary fractional Ornstein--Uhlenbeck process $\Yh{t}$ defined in \eqref{eq_Yh} satisfies the integrability assumption \eqref{assump_xi}.
\end{lem}
\begin{proof}
	This can be easily checked using the argument in \cite[Lemma~3.1]{FoHu:17} and the fact that $A^\eps(T) \equiv \int_0^T \MCK^\eps(s)\ud s$ is of order $\sqrt{\eps}$. Therefore, we omit the details here.
\end{proof}

\section{Merton problem under  fast-varying RFSE}\label{sec_asymppower}

In this section, we study the Merton problem \eqref{def_Vt} with  power utility, when the stochastic environment is modeled by $\Yh{t}$ with $H$ restricted to $H<\half$. By the nature of $\Yh{t}$, this nonlinear problem is non-Markovian. This immediately rules out the Hamilton-Jacobi-Bellman PDE approach, which is usually an efficient tool to analyze and find approximations to control problems with small parameters involved. Nevertheless, Proposition~\ref{prop_martdistort} is available and we will start with applying it to our problem and then find expansions based on it.

%In this section, we first introduce the $\eps$-scaled stationary fOU process denote by $\Yh{t}$, of which we mention several properties with proofs delayed to the Appendix. Then, we  study the Merton problem \eqref{def_Vt} under such fractional stochastic factor $\Yh{t}$. 

To be specific, we will give approximations of both the value process, denoted by $\Vy_t$ and the corresponding optimal strategy $\pi^\ast$. This is done by applying Proposition~\ref{prop_martdistort} with $Y_t = \Yh{t}$, then by expanding the expressions \eqref{def_Vy} based on the properties mentioned in Section~\ref{sec_fastfOU}. We also show that the ``leading order'' strategy alone can produce the given approximation of $\Vy_t$. 
%However, the implementation needs to track the fast factor $\Yh{t}$ using high-frequency data and this is not an easy task. To address this issue, we propose a practical strategy which does not require tracking $\Yh{t}$, with numerical illustration. 
Finally, we compare the results with the Markovian case, and we comment on the effects of taking into account the short-range dependence.

\subsection{First order approximation to the value process}\label{sec_asympVy}
Let $S_t$ be the price of a risky asset in the fast mean-reverting RFSE:
\begin{align}\label{def_StfastOU}
\ud S_t = S_t\left[\mu(\Yh{t})\ud t  + \sigma(\Yh{t})\ud W_t\right],
\end{align}
where $\Yh{t}$ is the $\eps$-scaled stationary fOU process with $H < \half$,  defined in \eqref{eq_Yh}. Accordingly, the wealth process $X_t^\pi$ becomes
\begin{equation}\label{def_XtunderfastY}
\ud X_t^\pi = \pi_t\mu(\Yh{t}) \ud t + \pi_t\sigma(\Yh{t}) \ud W_t,
\end{equation}
and the value process is denoted by $\Vy_t$:
\begin{equation}
\Vy_t \equiv \esssup_{\pi \in \MCA_t^\eps}\EE\left[U(X_T^\pi)\vert \MCF_t\right].
\end{equation}
We add the superscript $\eps$ to the problem value $\Vy_t$ and to the admissible set $\MCA_t^\eps$ to emphasize the dependence on $\eps$ brought by $\Yh{t}$.

Directly applying Proposition~\ref{prop_martdistort} with $Y_t = \Yh{t}$ gives the following expression for $\Vy_t$:
\begin{equation}\label{def_Vy}
\Vy_t = \frac{X_t^{1-\gamma}}{1-\gamma}\left[\widetilde \EE\left(e^{\frac{1-\gamma}{2q\gamma}\int_t^T \lambda^2(\Yh{s})\ud s} \Big\vert \MCG_t\right)\right]^q.
\end{equation}

\begin{theo}\label{thm_Vtpowerexpansion}
In the regime of $\eps$ small, under model assumptions, for fixed $t \in [0,T)$, $\Vy_t$ takes the form
\begin{equation}\label{eq_Vtpower}
\Vy_t = Q^\eps_t(t, X_t) + o(\sqrt\eps),
\end{equation}
where 
\begin{equation}\label{def_Qeps}
Q^\eps(t,x) = \vz(t,x) + \sqrt\eps \rho\overline D \vo(t,x),
% \frac{x^{1-\gamma}}{1-\gamma}e^{\frac{1-\gamma}{2\gamma}\overline\lambda^2(T-t)}\left[1  +\sqrt{\eps} \rho\frac{(1-\gamma)^2}{\gamma^2}\overline D (T-t)\right],
\end{equation}
with $\vz$ and $\vo$ defined as
\begin{equation}
\vz(t,x) \equiv \frac{x^{1-\gamma}}{1-\gamma}e^{\frac{1-\gamma}{2\gamma}\overline \lambda^2(T-t)}, \quad \text{and} \quad \vo(t,x) \equiv \frac{1-\gamma}{\gamma^2}(T-t)x^{1-\gamma}e^{\frac{1-\gamma}{2\gamma}\overline{\lambda}^2(T-t)},
\end{equation}
and the coefficient $\overline D$ defined by 
\begin{equation}\label{def_Dbar}
\overline D \equiv \int_0^{\infty}\left[\iint_{\RR^2} \lambda(\sigma_{ou} z)(\lambda\lambda')(\sigma_{ou}z') p_{C_Y(s)}(z,z') \ud z \ud z'\right]\mc{K}(s) \ud s.
\end{equation}
The function $p_C(z,z')$ is the pdf of the bivariate normal distribution with mean zero and covariance matrix $\begin{bmatrix}
1 &C\\
C & 1
\end{bmatrix}$, and $C_Y(s)$ is given in \eqref{eq_Zhcovar}. As usual, the notation $o(\sqrt{\eps})$ denotes an $\MCF_t$-adapted random variable whose order is higher than $\sqrt\eps$ in the $L^p$ sense, for any $1 \leq p < 2(1-H)$.
\end{theo}

\begin{proof}
Given the representation \eqref{def_Vy} for $\Vy$, the expansion result \eqref{eq_Vtpower}-\eqref{def_Qeps} can be obtained by firstly expanding 
\begin{equation}\label{def_Psi}
\Psi_t^\eps \equiv \widetilde \EE\left[e^{\frac{1-\gamma}{2q\gamma}\int_t^T \left(\lambda^2(\Yh{s})-\overline \lambda^2\right)\ud s}\Big\vert \MCG_t\right],
\end{equation}
and then applying Taylor formula to the function $x^q$. 

Using the fact that $I_t^\eps$ is ``small'' and Taylor expansion of $e^x$ in $x$, one deduces 
\begin{align}
\Psi_t^\eps
& = \widetilde\EE\left[ 1 + \frac{1-\gamma}{2q\gamma}\int_t^T \left(\lambda^2(\Yh{s})-\overline{\lambda}^2\right)\ud s + R_{[t,T]} \Big\vert\MCG_t \right]\nonumber\\
& = 1 + \frac{1-\gamma}{q\gamma}\widetilde\EE\left[\half\int_t^T \left(\lambda^2(\Yh{s})-\overline{\lambda}^2\right)\ud s \Big\vert\MCG_t \right] + \widetilde \EE\left[ R_{[t,T]} \vert \MCG_t\right]\label{eq_Psi},
\end{align}
where $
R_{[t,T]} = e^{\chi}\left[\frac{1-\gamma}{2q\gamma}\int_t^T \left(\lambda^2(\Yh{s}) - \overline{\lambda}^2\right) \ud s \right]^2$
with $\chi$ being the bounded Lagrange remainder. Observing that $\pnorm{\widetilde \EE[R_{[t,T]}\vert \MCG_t]}  \sim \pnorm{R_{[t,T]}}$ by the conditional H\"{o}lder inequality and the boundedness of $\lambda$, and that $\int_0^T \left(\lambda^2(\Yh{s}) - \overline \lambda^2\right) \ud s \sim o(\sqrt{\eps})$ in $L^2$, we claim that $\widetilde \EE\left[ R_{[t,T]} \vert \MCG_t\right]$ is of order higher than $\sqrt\eps$ in $L^p$ for $1 \leq p < 2(1-H)$.

Define the $\widetilde \PP$-martingale $\widehat \psi_t^\eps$ with its martingale representation by 
\begin{equation}\label{def_psihat}
\widehat\psi_t^\eps = \widetilde\EE\left[\int_0^T G(\Yh{s})\ud s \Big\vert\MCG_t \right],  \quad G(y) = \half(\lambda^2(y) - \overline{\lambda}^2), \quad \ud \widehat \psi_t^\eps = \widehat \vartheta_t^\eps \ud \widetilde W_t^Y.
\end{equation}
It remains to find the expansion of $\widehat \psi_t^\eps$ up to order $\sqrt{\eps}$.
To condense the notation in the following derivation, we also define
\begin{align}
&\psi_t^\eps \equiv \EE\left[\half\int_0^T \left(\lambda^2(\Yh{s}) - \overline\lambda^2\right) \ud s \Big\vert \MCG_t\right],\label{def_psi}\\
& \vartheta_t^\eps \equiv \int_t^T \EE\left[G'(\Yh{s})\vert\MCG_t\right]\kereps(s-t)\ud s, \label{def_vartheta}\\
&\widetilde \psi_t^\eps \equiv \widetilde \EE\left[\half\int_0^T \left(\lambda^2(\Yht{s}) - \overline\lambda^2\right) \ud s \Big\vert \MCG_t\right],\label{def_psitilde}\\
&\widetilde\vartheta_t^\eps \equiv \int_t^T \widetilde \EE[G'(\Yht{s})\vert\MCG_t]\kereps(s-t)\ud s, \label{def_varthettilde}
\end{align}
where $\psi_t^\eps$ is a $\PP$-martingale satisfying $\ud \psi_t^\eps = \vartheta_t^\eps \ud W_t^Y$, see Lemma~\ref{lem_moments}\eqref{lem_psi} for details. By a similar argument, we have the $\widetilde \PP$-martingale $\widetilde \psi_t^\eps$ admitting the representation $\ud \widetilde \psi_t^\eps = \widetilde\vartheta_t^\eps \ud \widetilde W_t^Y$. The difference between $ \vartheta_t^\eps$ and $ \widetilde \vartheta_t^\eps$ is of order $\eps$ and this is discussed in Lemma~\ref{lem_comparison}\eqref{lem_varthetatilde}.

With all above preparations, we deduce:
\begin{align*}
\widehat \psi_t^\eps &= \widetilde \EE\left[\int_0^T G(\Yh{s})\ud s \Big \vert \MCG_0\right] + \int_0^t \widehat \vartheta_s^\eps \ud \widetilde W_s^Y \\
&\hspace{300pt}(\text{Taylor expanding } G(y) \text{ at } y = \Yht{s})\\
&  = \widetilde \EE\left[\int_0^T G(\Yht{s})\ud s \Big \vert \MCG_0\right] +  \widetilde \EE\left[\int_0^T G'(\Yht{s})(\Yh{s} - \Yht{s})\ud s \Big \vert \MCG_0\right]  \\
&\quad  +  \widetilde \EE\left[\int_0^T G''(\chi_s)(\Yh{s} -\Yht{s})^2\ud s \Big \vert \MCG_0\right] + \int_0^t \vartheta_s^\eps \ud \widetilde W_s^Y +  \int_0^t (\widehat \vartheta_s^\eps -  \vartheta_s^\eps) \ud \widetilde W_s^Y \\
&\hspace{120pt} (\Yh{s} - \Yht{s} \sim O(\sqrt\eps), \widehat \vartheta_s^\eps -  \vartheta_s^\eps \sim O(\eps) \text{ and the relation between } W^Y \text{ and } \widetilde W^Y)\\
& = \widetilde \EE\left[\int_0^T G(\Yht{s})\ud s \Big \vert \MCG_0\right] +  \widetilde \EE\left[\int_0^T G'(\Yht{s})\int_0^s \MCK^\eps(s-u)\rho\left(\frac{1-\gamma}{\gamma}\right)\lambda(\Yh{u})\ud u\ud s \Big \vert \MCG_0\right] \\
& \quad + \int_0^t \vartheta_s^\eps \ud W_s^Y -  \int_0^t \vartheta_s^\eps \rho\left(\frac{1-\gamma}{\gamma}\right)\lambda(\Yh{s}) \ud s  + \MCO(\eps) \\
&\hspace{275pt} (\Yht{s}\vert \MCG_0 \stackrel{\MCD}{=} \Yh{s}\vert \MCG_0 \text{ and the definition of } \kappa_t^\eps)
\end{align*}
\begin{align*}
\hspace{11pt}& = \EE\left[\int_0^T G(\Yh{s})\ud s \Big \vert \MCG_0\right] +  \int_0^t \vartheta_s^\eps \ud W_s^Y + \rho\left(\frac{1-\gamma}{\gamma}\right)\widetilde \EE\left[\int_0^T \int_u^T \widetilde \EE[G'(\Yht{s})\vert\MCG_u] \MCK^\eps(s-u)\ud s \lambda(\Yh{u})\ud u \Big \vert \MCG_0\right] \\
& \quad - \rho\left(\frac{1-\gamma}{\gamma}\right)( \sqrt{\eps}\,\overline D t + \kappa_t^\eps) + \MCO(\eps)\\
& \hspace{250pt}(\text{expression of } \psi_t^\eps \text{ and } \widetilde\vartheta_t^\eps, \text{ and } \ltwonorm{\kappa_t^\eps} \sim o(\sqrt{\eps}))\\
& = \psi_t^\eps +\rho\left(\frac{1-\gamma}{\gamma}\right)\widetilde \EE\left[\int_0^T \widetilde \vartheta_u^\eps \lambda(\Yh{u})\ud u \Big \vert \MCG_0\right] -  \rho\left(\frac{1-\gamma}{\gamma}\right) \sqrt{\eps}\,\overline D t + o(\sqrt{\eps}) \\
& \hspace{350pt}(\vartheta_t^\eps - \widetilde \vartheta_t^\eps \sim \MCO(\eps))\\
& = \psi_t^\eps +\rho\left(\frac{1-\gamma}{\gamma}\right)\widetilde \EE\left[\int_0^T  \vartheta_u^\eps \lambda(\Yh{u})\ud u \Big \vert \MCG_0\right] -  \rho\left(\frac{1-\gamma}{\gamma}\right) \sqrt{\eps}\,\overline D t + o(\sqrt{\eps}) \\
& =  \psi_t^\eps +  \rho\left(\frac{1-\gamma}{\gamma}\right) \sqrt{\eps}\,\overline D(T-t) +\rho\left(\frac{1-\gamma}{\gamma}\right)\widetilde \EE\left[\int_0^T  \vartheta_u^\eps  \lambda(\Yh{u}) - \sqrt{\eps} \, \overline D\ud u \Big \vert \MCG_0\right] +  o(\sqrt{\eps})\\
&\hspace{350pt} (\ltwonorm{\kappa_t^\eps} \sim o(\sqrt{\eps})) \\
& =  \psi_t^\eps +  \rho\left(\frac{1-\gamma}{\gamma}\right) \sqrt{\eps}\,\overline D(T-t) +  o(\sqrt{\eps}).
\end{align*}
All reasonings are mentioned in the parentheses from line to line and detailed statements can be found in Lemmas~\ref{lem_moments}--\ref{lem_psihat}. Subtracting $\int_0^t G(\Yh{u}) \ud u$ from both sides of the above expansion, together with \eqref{eq_Psi}, \eqref{def_phi} and \eqref{def_psi}, brings
\begin{align}
\Psi_t^\eps &= 1 + \frac{1-\gamma}{q\gamma} \left( \phi_t^\eps  + \sqrt\eps\rho\left(\frac{1-\gamma}{\gamma}\right) \overline D (T-t)\right) + o(\sqrt\eps) \nonumber\\
& =  1 + \frac{1-\gamma}{q\gamma} \sqrt\eps\rho\left(\frac{1-\gamma}{\gamma}\right) \overline D (T-t) + o(\sqrt\eps).
\label{eq_Psiexpansion}
\end{align}
The last step follows from  $\phi_t^\eps \sim \MCO(\eps^{1-H})$ (see Lemma~\ref{lem_moments}\eqref{lem_phi}). Now, Taylor expanding $x^q$ produces the desired result
\begin{align*}
\Vy_t&= \frac{X_t^{1-\gamma}}{1-\gamma}e^{\frac{1-\gamma}{2\gamma}\overline{\lambda}^2(T-t)} \left(\Psi^\eps_t\right)^q \\
& =\frac{X_t^{1-\gamma}}{1-\gamma}e^{\frac{1-\gamma}{2\gamma}\overline{\lambda}^2(T-t)} \left\{1 +  \sqrt\eps\rho\left(\frac{1-\gamma}{\gamma}\right)^2\overline D (T-t)
\right\} + o(\sqrt{\eps}),
\end{align*}
where $o(\sqrt{\eps})$ is in $L^p$ sense, for any $1 \leq p < 2(1-H)$.
\end{proof}

Note that, unlike in the long-range dependent $H > 1/2$ case studied in \cite{FoHu2:17} where  the first order corrections consist of two terms (one random component $\phi_t^\eps$ and one deterministic function in $(t, X_t)$) at order $\eps^{1-H}$, 
here, the correction appears at order $\sqrt{\eps}$ and contains only a deterministic function of $(t, X_t)$. In other words, except for the constant $\overline{D}$ and $\overline{\lambda}$, neither the fast factor $\Yh{t}$ nor the Hurst index $H$ is visible in the leading order term nor in the first order correction.

\subsection{First order expansion of the optimal strategy}\label{sec_asymppi}
The optimal portfolio  $\pi^\ast$ is also of interest, if not the most important quantity. Under the RFSE described by $\Yh{t}$, using the results in Proposition~\ref{prop_martdistort}, the optimal strategy \eqref{def_pioptimal}  takes the form
\begin{equation}\label{def_pioptimalunderfOU}
\pi^\ast_t = \left[\frac{\lambda(\Yh{t})}{\gamma \sigma(\Yh{t})} + \frac{\rho q \xi_t}{\gamma \sigma(\Yh{t})}\right] X_t.
\end{equation}
The term $\xi_t$ is defined by the martingale representation of $\widetilde \PP$-martingale $M_t$, and is in general not known explicitly. This brings extra difficulty when one wants to implement the optimal strategy $\pi^\ast$ to attain the problem value. However, at least in the regime of $\eps$ small, we can give the following approximation result for $\xi$ and $\pi^\ast$.
\begin{theo}\label{thm_piexpansion}
Under model assumptions, we have the following approximation of the optimal strategy $\pi_t^\ast$:
{
\begin{equation}
\left(\EE\int_0^T \abs{\pi^\ast_t - \pi_t^{(0)}}^p \ud t \right)^{1/p} \sim o(\sqrt{\eps}), \quad \forall 1 \leq p < 2(1-H),
\end{equation}
where $\pi_t^{(0)}$ is the leading order strategy:
\begin{equation}\label{def_pz}
\pi_t^{(0)} := \frac{\lambda(\Yh{t})}{\gamma \sigma(\Yh{t})}X_t. 
\end{equation}}
\end{theo}

\begin{proof}
This is done by obtaining the expansion of $\xi_t$ from its definition \eqref{def_xi}. We rewrite $M_t$ in terms of $\Psi_t^\eps$ by comparing \eqref{def_Mtmartdistort} to \eqref{def_Psi}, 
\begin{equation*}
M_t = \Psi_t^\eps \; e^{\frac{1-\gamma}{2q\gamma}\int_0^t \lambda^2(\Yh{s})\ud s} \; e^{\frac{1-\gamma}{2q\gamma}\overline{\lambda}^2(T-t)},
\end{equation*}
and then,  we use the approximation \eqref{eq_Psiexpansion} of $\Psi_t^\eps$. {It has been shown that 
\begin{align*}
\Psi_t^\eps = a_\eps(t) + R_\eps(t),
\end{align*}
where $a_\eps(t) = 1 + \frac{1-\gamma}{q\gamma}\sqrt{\eps}\rho\left(\frac{1-\gamma}{\gamma}\right)\overline D (T-t)$ is of finite variation, and $R_\eps(t)$ is of order $o(\sqrt{\eps})$ in $L^p$, for $1 \leq p < 2(1-H)$. This ensures that $\average{R_\eps(\cdot), W(\cdot)}_t \sim o(\sqrt{\eps})$ in $L^p$ by the following reasoning. Using $L^2$ representation theorem, for each $t \in [0,T]$, there exists an adapted process $\beta(\cdot)$, such that $\EE\int_0^t \beta^2(u) \ud u < \infty$ and $R_\eps(t)$ admits the representation
\begin{equation*}
R_\eps(t) = \EE[R_\eps(t)] + \int_0^t \beta(u) \ud W_u.
\end{equation*}
Since $\abs{\EE[R_\eps(t)]} \leq \left(\EE\abs{R_\eps(t)}^p\right)^{1/p} \sim o(\sqrt{\eps})$, we deduce
\begin{equation*}
\EE\abs{\int_0^t \beta(u)\ud W_u}^p = \EE\abs{R_\eps(t) - \EE[R_\eps(t)]}^p \leq C (\EE\abs{R_\eps(t)}^p + \EE^p\abs{R_\eps(t)}) \sim o(\sqrt{\eps}^p),
\end{equation*}
and 
\begin{equation*}
\EE\abs{\average{R_\eps, W}_t}^p \leq \EE\int_0^t \abs{\beta(u)}^p \ud u \leq C \EE\abs{\int_0^t \beta(u)\ud W_u}^p \leq o(\sqrt{\eps}^p).
\end{equation*}

Now applying It\^o's formula to the above expression of $M_t$, none of the three terms will contribute to the diffusion part at order $\sqrt{\eps}$, meaning that $\int_0^t M_u\xi_u \ud u \sim o(\sqrt{\eps})$ in $L^p$, for every $t \in [0,T]$. By model assumptions, $M_t$ is bounded, therefore  
\begin{equation*}
\int_0^t \xi_u \ud u \sim o(\sqrt{\eps}) \text{ in } L^p, \quad \forall t \in [0,T].
\end{equation*}

Denote by $\pnorm{\xi} := \left(\EE\int_0^T \abs{\xi_t}^p \ud t\right)^{1/p}$ the $L^p(\Omega\times[0,T])$ norm, together with the estimates of $\xi_t$ in Lemma~\ref{lem_xi}, that is, $\abs{\xi} \leq C\sqrt{\eps}$ for any $t \in [0,T]$, we have the desired approximation \eqref{def_pz}:
\begin{align*}
\left(\EE\int_0^T \abs{\pi^\ast_t - \pi_t^{(0)}}^p \ud t \right)^{1/p} &= \pnorm{\pi^\ast - \pi^{(0)}} \leq C \pnorm{\xi X}  = C \prnorm{\xi}\pqnorm{X} \\
& \leq C'\left(\pnorm{\xi}(\sqrt{\eps})^{r - 1}\right)^{1/r}, \text{ where } \frac{1}{r} + \frac{1}{q} = 1 \text{ and }r, q >1\\
& \sim o(\sqrt{\eps}).
\end{align*}
}

\end{proof}

%Note that, in the above approximation, both the leading order strategy $\pi_t^{(0)}$ and the first order correction term $\pi_t^{(1)}$ are in feedback forms in terms of the state processes. Therefore, if one decides to track the fast-varying process $\Yh{t}$ to implement $\pi_t^{(0)}$, no further computational cost is required when $\pi_t^{(1)}$ is also included in order to incorporate the inter-temporal hedging. On the other hand, tracking $\Yh{t}$ is not easy and requires sophisticated econometric techniques. This issue will be addressed in Section~\ref{sec_practicalpz}. Before that, we discuss how good the strategy $\pi_t^{(0)}$ is.

In the next subsection, we show the asymptotic optimality property of $\pz$ defined in \eqref{def_pz}. That is, by only implementing $\pz$, the agent is able to obtain the first order approximation \eqref{def_Qeps} of the optimal value \eqref{def_Vy}.

\subsection{Asymptotic optimality of $\pi_t^{(0)}$}\label{sec_asymppzpower}
%%In this subsection, we investigate the relation between $\Vy_t$ and the value function obtained by following the zeroth-order strategy given in \eqref{eq_piapprox}:
%\begin{equation*}
%\pi_t^{(0)} = \frac{\lambda(\Yh{t})}{\gamma\sigma(\Yh{t})}X_t. \label{def_pzpower}
%\end{equation*}
Let $X_t^\pz$ be the wealth process associated to $\pi_t^{(0)}=  \frac{\lambda(\Yh{t})}{\gamma\sigma(\Yh{t})}X_t$:
%following to the zeroth order strategy $\pi_t^{(0)} = \frac{\lambda(\Yh{t})}{\gamma\sigma(\Yh{t})}X_t^\pz$ obtained in \eqref{eq_piapprox} , then $X_t^\pz$ satisfies 
\begin{align}
\ud X_t^\pz &= \mu(\Yh{t})\pi_t^{(0)} \ud t + \sigma(\Yh{t})\pi_t^{(0)} \ud W_t \nonumber\\
& = \frac{\lambda^2(\Yh{t})}{\gamma}X_t^\pz \ud t + \frac{\lambda(\Yh{t})}{\gamma}X_t^\pz \ud W_t.\label{def_Xtpz}
\end{align}
By the boundedness of $\lambda(\cdot)$, $X_t^\pz$ has $p^{th}$-moment for any $p$, and this ensures the admissibility of $\pz$. 

To systematically simplify the notation in the derivation in Proposition~\ref{cor_optimalityofpz}, we introduce the risk-tolerance function $R(t,x)$ and the differential operator $D_k$, as in \cite{FoSiZa:13}, by:
\begin{equation}\label{def_risktolerance}
R(t,x) \equiv - \frac{\vz_x(t,x)}{\vz_{xx}(t,x)} = \frac{x}{\gamma}, \quad \text{and} \quad D_k \equiv R(t,x)^k \partial_x^k,\; \forall k \in \NN^+.
\end{equation}
Then, the wealth process $X_t^\pz$ can be written as
\begin{align}\label{def_Xt}
\ud X_t^\pz = \lambda^2(\Yh{t}) R(t, X_t^\pz) \ud t + \lambda(\Yh{t})R(t,X_t^\pz) \ud W_t.
\end{align}
We also introduce the nonlinear operator $\Ltx(\lambda)$ by
\begin{equation}
\Ltx(\lambda) \equiv \partial_t + \frac{1}{2}\lambda^2D_2 + \lambda^2D_1,\label{def_ltx}
\end{equation}
and one can check by straightforward calculation that $\vz$ satisfies
\begin{align}\label{eq_vz}
&\Ltx(\overline\lambda)\vz(t,x) = 0.
\end{align}

Denote by $\Vyl_\cdot$ the corresponding value process 
\begin{equation*}
\Vyl_t:= \EE\left[\left.U\left( X_T^\pz \right)\right\vert \MCF_t\right].
\end{equation*}
In what follows, we aim to find the approximation of $\Vyl_t$ in the regime of $\eps$ small. This will be obtained via ``epsilon-martingale decomposition'', firstly introduced in \cite{FoPaSi:00} to solve the linear pricing problem, and later developed in \cite{FoPaSi:01,GaSo:15,GaSo:16,GaSo:17,FoHu:17,FoHu2:17}. Roughly speaking, we need to construct an explicit function $Q_t$ for $\Vyl_t$ in the form of a martingale plus something small (non-martingale part), which has the same terminal condition as $\Vyl_t$. Then, this ansatz is indeed the approximation to $\Vyl_t$ up to order of the non-martingale part. Detailed explanation can be found in the above references.  

\begin{prop}\label{cor_optimalityofpz}
Under model assumptions, for fixed $t \in[0,T)$ and the  observed value $X_t$, $\Vyl_t$ is approximated by
	\begin{equation}\label{def_Vylpower}
	\Vyl_t = Q_t^\eps(X_t) + o(\sqrt{\eps}),
	\end{equation}
where $Q_t^\eps$ is given in \eqref{def_Qeps}.
\end{prop}
The above Proposition combined with Theorem~\ref{thm_Vtpowerexpansion} immediately gives:

\centerline{\it
	$\pi_t^{(0)}$ is asymptotically optimal within all admissible strategy $\MCA_t^\eps$ up to order $\sqrt\eps$.}

\noindent This is because $\Vyl_t- \Vy_t$ is of order $o(\sqrt{\eps})$, which indicates that  $\pi_t^{(0)}$
already generates the leading order term plus the correction of order $\sqrt\eps$ given by \eqref{def_Qeps}.

\begin{proof}[Proof of Proposition~\ref{cor_optimalityofpz}]
	
		Based on the epsilon-martingale decomposition approach, it is enough to find a decomposition $M_t^\eps + R_t^\eps$ for 
		$Q_t^{\eps}$, with $M_t^\eps$ being a true martingale, and $R_t^\eps$ being of order $o(\sqrt\eps)$. We present how $M_t^\eps$ and $R_t^\eps$ are determined, with actual  proofs delayed to Appendix~\ref{app_lemmas}. Note that terms of order $\sqrt{\eps}$ are included in $M_t^\eps$ so that $R_t^\eps$ is pushed to a higher order.
	
		Applying It\^{o} formula to $\vz(t,X_t^\pz)$ brings 
		\begin{align}
		\ud \vz(t, X_t^\pz) &= \Ltx(\lambda(\Yh{t}))\vz(t,X_t^\pz) \ud t + \sigma(\Yh{t})\pz(t,X_t^\pz, \Yh{t})\vz_x(t,X_t^\pz) \ud W_t \nonumber\\
		& = \half \left(\lambda^2(\Yh{t}) - \overline{\lambda}^2\right) D_1\vz(t,X_t^\pz) \ud t + \ud M_t^{(1)},\label{eq_s1}
		\end{align}
		where $M_t^{(1)}$ is the martingale given by
		\begin{equation}\label{def_M1}
		\ud M_t^{(1)} = \sigma(\Yh{t})\pz(t,X_t^\pz, \Yh{t})\vz_x(t,X_t^\pz) \ud W_t,
		\end{equation}
		and the relations \eqref{eq_vz} and  $D_1\vz(t,x) = -D_2\vz(t,x)$ have been used.
		
		Recall $\phi_t^\eps$ and $\psi_t^\eps$ defined in \eqref{def_phi} and \eqref{def_psi} respectively, then, we have
		$\ud \psi_t^\eps - \ud \phi_t^\eps = \half\left(\lambda^2(\Yh{t}) - \overline{\lambda}^2\right)\ud t$, and 
		the first term in \eqref{eq_s1} becomes
		\begin{equation*}
		\half \left(\lambda^2(\Yh{t}) - \overline{\lambda}^2\right) D_1\vz(t,X_t^\pz) \ud t = D_1\vz(t,X_t^\pz) \left(\ud \psi_t^\eps - \ud \phi_t^\eps\right).
		\end{equation*}
		To further simplify $D_1\vz(t,X_t^\pz)\ud \phi_t^\eps$, which corresponds to finding the corrector to $\vz(t,X_t^\pz)$ at order $\sqrt\eps$, we compute the total differential of $D_1\vz(t,X_t^\pz)\phi_t^\eps$ (the arguments of $\vz(t,X_t^\pz)$ will be omitted systematically in the following):
		\begin{align*}
		\ud \left(D_1\vz\phi_t^\eps\right) &= D_1\vz\ud \phi_t^\eps + \phi_t^\eps\Ltx(\lambda(\Yh{t})) D_1\vz \ud t + \phi_t^\eps \sigma(\Yh{t})\pz(t,X_t^\pz,\Yh{t})\partial_x D_1\vz \ud W_t \\
		&\hspace{10pt}+ \sigma(\Yh{t})\pz(t,X_t^\pz, \Yh{t}) \partial_x D_1\vz \ud \average{W,\phi^\eps}_t \\
		&= D_1\vz\ud \phi_t^\eps + \phi_t^\eps\left[\half(\lambda^2(\Yh{t})-\overline\lambda^2)(D_2 + 2D_1)D_1\vz\right]\ud t   \\
		&\hspace{10pt}+ \phi_t^\eps \lambda(\Yh{t})D_1^2\vz \ud W_t + \rho\lambda(\Yh{t})D_1^2\vz \ud \average{W^Y,\psi^\eps}_t
		\end{align*}
		In the derivation, we have used the definition of $D_1$ and $R(t,x)$ (cf. \eqref{def_risktolerance}), and
		\begin{equation*}
		\Ltx(\overline{\lambda})D_1\vz = D_1\Ltx(\overline{\lambda})\vz = 0, \text{ and } \ud \average{W, \phi^\eps}_t = \rho\ud \average{W^Y, \psi^\eps}_t.
		\end{equation*}
		The results in Lemma~\ref{lem_moments}\eqref{lem_psi}: $\ud \average{W^Y, \psi^\eps}_t = \vartheta_t^\eps\ud t$, together with the above derivation produce
		\begin{align}
		\ud \left(D_1\vz\phi_t^\eps\right) &= -\half \left(\lambda^2(\Yh{t}) - \overline{\lambda}^2\right) D_1\vz \ud t +  \phi_t^\eps\left[\half(\lambda^2(\Yh{t})-\overline\lambda^2)(D_2+2D_1)D_1\vz\right]\ud t \nonumber\\
		& \hspace{10pt}+ \rho\lambda(\Yh{t})D_1^2\vz \vartheta_t^\eps \ud t + \ud M_t^{(2)},\label{eq_s2}
		\end{align}
		and 
		\begin{align}\label{def_M2}
		\ud M_t^{(2)} = D_1\vz(t, X_t^\pz) \ud \psi_t^\eps + \phi_t^\eps \lambda(\Yh{t})D_1^2\vz(t, X_t^\pz) \ud W_t.
		\end{align}
		From Lemma~\ref{lem_moments}\eqref{lem_phi}, we see that the second term in \eqref{eq_s2} is of order $\eps^{1-H}$, and thus it remains to analyze the third term $\rho\lambda(\Yh{t})D_1^2 \vz \vartheta_t^\eps \ud t$.  To this end, we rewrite $\vo$ in terms of $D_1$ and $\vz$ by $\vo = D_1^2\vz(T-t)$, and observe that it satisfies $\Ltx(\overline\lambda)\vo = -D_1^2 \vz$. Apply It\^{o}'s formula again to $\vo$ gives,
			\begin{align}
			\ud \vo(t,X_t^\pz) &= \Ltx(\lambda(\Yh{t})) \vo(t, X_t^\pz) \ud t + \sigma(\Yh{t})\pz(t,X_t^\pz, \Yh{t})\vo_x(t,X_t^\pz) \ud W_t\nonumber\\
			& = \half(\lambda^2(\Yh{t}) - \overline{\lambda}^2)(D_2 + 2D_1)\vo(t, X_t^\pz) \ud t - D_1^2\vz(t,X_t^\pz)\ud t + \ud M_t^{(3)},\label{eq_s3}
			\end{align}
			where $M_t^{(3)}$ is the martingale defined by
			\begin{equation}\label{def_M3}
			\ud M_t^{(3)} = \sigma(\Yh{t})\pz(t,X_t^\pz, \Yh{t})\vo_x(t,X_t^\pz) \ud W_t.
			\end{equation}

%		Using results in Lemma~\eqref{lem_psi}, we have $\lambda(\Yh{t})\vartheta_t^\eps = \sqrt{\eps} \,\overline D$
%		
%		The term $\rho\lambda(\Yh{t})D_1^2\vz \theta_t\ud t$ is taken care of by adding the term $\eps^{1-H}\rho\widetilde{\lambda}\vo$ to $Q_t^{\pz, \eps}$. By using the relation $\theta_t = -\partial_t C_{t,T}$, one has
	
		Defining the quantity $\widetilde Q_t^\eps$ by
		\begin{equation}\label{def_Qtilde}
		\widetilde Q_t^\eps(x) = \vz(t,x) + D_1\vz(t,x)\phi_t^\eps + \sqrt{\eps} \rho \overline D \vo(t,x),
		\end{equation}
		and combining equation \eqref{eq_s1}, \eqref{eq_s2} and \eqref{eq_s3} yields
		\begin{align*}
		\ud \widetilde Q_t^{\eps}(X_t^\pz) &= \ud \left(\vz(t, X_t^\pz) + D_1\vz(t, X_t^\pz) \phi_t^\eps +  \sqrt\eps\rho\overline D\vo(t, X_t^\pz)\right)\\
		& = \phi_t^\eps\left[\half(\lambda^2(\Yh{t})-\overline\lambda^2)(D_2+2D_1)D_1\vz\right]\ud t  +  \rho\left(\lambda(\Yh{t})\vartheta_t^\eps-\sqrt{\eps}\,\overline D\right)D_1^2\vz \ud t \\
		& \hspace{10pt} +  \half\sqrt\eps\rho\overline D (\lambda^2(\Yh{t}) - \overline{\lambda}^2)(D_2 + 2D_1)\vo(t, X_t^\pz) \ud t  \\
		& \hspace{10pt}+ \ud M_t^{(1)} + \ud M_t^{(2)} + \sqrt\eps\rho\overline D\ud M_t^{(3)}.
		\end{align*}
		Denote by $R_{t,T}^{(j)}$, $j=1, 2, 3$ the first three terms in the above expression
		\begin{align}
		&R_{t,T}^{(1)} := \int_t^T  \phi_s^\eps\left[\half(\lambda^2(\Yh{s})-\overline\lambda^2)(D_2+2D_1)D_1\vz(s,X_s^\pz)\right]\ud s, \label{def_R1general}\\
		&R_{t,T}^{(2)} := \int_t^T  \rho\left(\lambda(\Yh{s})\vartheta_s^\eps-\sqrt{\eps}\,\overline D \right)D_1^2\vz(s, X_s^\pz)\ud s,\label{def_R2general}\\
		&R_{t,T}^{(3)} := \int_t^T  \half\sqrt\eps\rho\overline D (\lambda^2(\Yh{s}) - \overline{\lambda}^2)(D_2 + 2D_1)\vo(s, X_s^\pz) \ud s.\label{def_R3general}
		\end{align}
		It is proved in Lemma~\ref{lem_Rjgeneral} that they are $o(\sqrt\eps)$ terms in $L^1$:
		\begin{equation}\label{eq_Rjgeneral}
		\lim_{\eps\to 0}\eps^{-1/2}\; \EE\abs{R_{t,T}^{(j)}} = 0, \quad \forall j = 1, 2, 3.
		\end{equation}
		Lemma~\ref{lem_Mj} also shows that $M_t^{(j)}$, $j=1, 2, 3$ are indeed true $\PP$-martingales.
		
		Therefore, define the martingale $M_t^\eps$ and the non-martingale part $R_t^\eps$ respectively by
		\begin{align*}
		&M_t^\eps := \int_0^t  \ud M_s^{(1)} + \ud M_s^{(2)} + \sqrt\eps\rho\overline D\ud M_s^{(3)},\\
		&R_{T}^\eps - R_t^{\eps} := R_{t,T}^{(1)} + R_{t,T}^{(2)} +R_{t,T}^{(3)},
		\end{align*}
		and observe that $\widetilde Q_T^{\eps}(x) = \vz(T,x) = U(x)$ (since $\phi_T^\eps = \vo(T,x) = 0$ by definition), and $D_1\vz(t,X_t^\pz)\phi_t^\eps$ is of order $o(\sqrt{\eps})$ in $L^1$ (by Lemma~\ref{lem_moments}\eqref{lem_phi} and integrability of $D_1\vz$), we obtain the desired result
		\begin{align*}
		\Vyl_t &= \EE\left[\widetilde Q_T^{\eps}(X_T^\pz)\big\vert \MCF_t\right] = \widetilde Q_t^{\eps}(X_t^\pz) + \EE[M_T^\eps - M_t^\eps\vert \MCF_t] + \EE[R_T^\eps - R_t^\eps \vert\MCF_t]\\
		& = Q_t^{\eps}(X_t^\pz) + D_1\vz(t,X_t^\pz)\phi_t^\eps +  \EE[R_{t,T}^{(1)} + R_{t,T}^{(2)} +R_{t,T}^{(3)} +R_{t,T}^{(4)} \vert\MCF_t] \\
		& = Q_t^{ \eps}(X_t^\pz) + o(\sqrt\eps).
		\end{align*}
		
	\end{proof}

\begin{rem}
Expansion results of $\Vyl$ can be extended to the case with general utility functions, as in \cite[Section~4]{FoHu:17,FoHu2:17}. This is accomplished using the properties of the risk-tolerance function $R(t,x)$ studied in \cite{FoHu:16}.
\end{rem}
 %Note that  it is important to get $\pi_t^{(0)}$ correctly. For one thing, it can generate the value process $\Vy_t$ not only to the leading order term, but also to the first order $\eps^{1-H}$ corrections, as discussed above; for another, in the valid regime $H \in (\half, 1)$, $\eps^{1-H}$ terms in $\Vy_t$ are relatively large, and are recommended to be included. Moreover, it is also recommended to include $\pi_t^{(1)}$ if one wants to get corrections of higher order than $\MCO(\eps^{1-H})$ in $\Vy_t$. 

\subsection{Numerical Implementation}
Numerical implementing $\pz$ needs to track $\Yh{t}$. By the high oscillation nature of $\Yh{t}$, this usually requires high frequency data and to deal with micro-structure issues. This is not practical for long-period investment and agents usually prefer not to tackle this issue. Instead, they would look into strategies which do not depend on the factor $\Yh{t}$. As shown in \cite{FoHu2:17}, such a strategy takes the form of classical Merton optimal strategies under Black-Scholes setting with drift $\overline \lambda$ and volatility $\overline \sigma: = \sqrt{\average{\sigma^2}}$:
\begin{equation}\label{def_pzbar}
 \bar\pi_t^{(0)} = \frac{\overline{\mu}}{\gamma\overline{\sigma}^2} X_t.
 \end{equation}
 
 To measure the utility loss of using $\bar{\pi}_t^{(0)}$, we define the associated problem value $V_t^{\bar\pi^{(0)},\eps}$
 \begin{equation}\label{def_Vtpibar}
  V_t^{\bar\pi^{(0)}, \eps} = \EE[U(X_T^{\bar\pi^{(0)}})\vert \MCF_t].
 \end{equation}
 Using the ergodic property of $\Yh{t}$:
 \begin{equation*}
 \int_t^T \left(\mu(\Yh{s}) - \overline{\mu}\right) \ud s \sim o(1), \quad \text{and} \quad \int_t^T \left(\sigma^2(\Yh{s})-\overline{\sigma}^2\right) \ud s \sim o(1),
 \end{equation*}
one can deduces the optimal leading order term
 \begin{equation*}
 \frac{X_t^{1-\gamma}}{1-\gamma}e^{\frac{1-\gamma}{2\gamma}\frac{\overline{\mu}^2}{\overline{\sigma}^2}(T-t)}.
 \end{equation*}
 This can be interpreted as the optimal value with Sharpe ratio $\overline{\mu}/\overline{\sigma}$. Then, the principal term  of utility loss of using $\bar\pi_t^{(0)}$ is quantified by comparing the above term with the leading order term of $\Vy_t$ given in \eqref{eq_Vtpower}-\eqref{def_Qeps}:
 \begin{equation*}
 \frac{X_t^{1-\gamma}}{1-\gamma}e^{\frac{1-\gamma}{2\gamma}\overline{\lambda}^2(T-t)},
 \end{equation*}
 and is measured by the Cauchy-Schwarz gap
 \begin{equation*}
 \overline{\lambda}^2 = \average{\frac{\mu^2}{\sigma^2}} \geq \frac{\average{\mu^2}}{\average{\sigma^2}}  \geq \frac{\overline{\mu}^2}{\overline{\sigma}^2},
 \end{equation*}
 as in the Markovian setup in \cite{FoSiZa:13}, and in the long-memory fractional setup in \cite{FoHu2:17}.

Next, we illustrate numerically the asymptotic optimality property of $\pi_t^{(0)}$ and the sub-optimality of $\bar{\pi}_t^{(0)}$. To this end,  we numerically evaluate $\Vy_t$, $\Vyl_t$,  and $V_t^{\bar\pi^{(0)}, \eps }$  at time $t=0$, and compare their differences. Applying a change of measure to equation \eqref{def_Vy}, we have
\begin{align*}
\Vy_0 =  \frac{X_0^{1-\gamma}}{1-\gamma}\left[\EE\left(e^{\left(\frac{1-\gamma}{2\gamma}\right)\int_0^T \lambda^2(\Yh{s})\ud s + \rho\left(\frac{1-\gamma}{\gamma}\right)\int_0^T \lambda(\Yh{s})\ud W_s^Y} \Big\vert \MCG_0\right)\right]^q.
\end{align*} 
Solving the SDE for $X_t^\pz$ and plugging the solution into the definition of $\Vyl_t$ bring
\begin{align*}
\Vyl_0 = \frac{X_0^{1-\gamma}}{1-\gamma}\EE\left(e^{\left(\frac{-2\gamma^2 + 3\gamma - 1}{2\gamma^2}\right)\int_0^T \lambda^2(\Yh{s})\ud s + \left(\frac{1-\gamma}{\gamma}\right)\int_0^T \lambda(\Yh{s})\ud W_s} \Big\vert \MCF_0\right).
\end{align*}
Similarly, the value process defined in \eqref{def_Vtpibar} following the $Y$-independent strategy $\bar\pi_t^{(0)}$ is given by
\begin{equation*}
V_0^{\bar\pi^{(0)}, \eps} = \frac{X_0^{1-\gamma}}{1-\gamma} \EE\left(e^{\left(\frac{1-\gamma}{\gamma}\right)\frac{\overline{\mu}}{\overline{\sigma}^2}\int_0^T \mu(\Yh{s})\ud s - \left(\frac{1-\gamma}{2\gamma^2}\right)\frac{\overline{\mu}^2}{\overline{\sigma}^4}\int_0^T \sigma^2(\Yh{s})\ud s  + \left(\frac{1-\gamma}{\gamma}\right)\frac{\overline{\mu}}{\overline{\sigma}^2}\int_0^T \sigma(\Yh{s})\ud W_s} \Big\vert \MCF_0\right).
\end{equation*}

A full study of the quality of the approximation to optimal strategy involves precise simulations of the fractional OU process $\Yh{t}$ and Brownian motion $W_t$ jointly, and mesh-size determination in $t$ so that the numerical error is negligible comparing to the loss in value function given by the approximated strategy. This is beyond the scope of this paper, and we therefore compute the values for only a few ``omegas'' for a purpose of illustration. 

The model parameters are chosen as:
\begin{equation*}
T = 1, \quad H = 0.1, \quad a = 1, \quad  \gamma = 0.4, \quad \rho = -0.5, \quad \mu(y) = \frac{0.1 \times \lambda(y)}{0.1 + \lambda(y)}, \quad  \lambda^2(y) = \half \int_{-\infty}^{y/\sigma_{ou}}p(z/2)\ud z,
\end{equation*}
where $p(z)$ denotes the standard normal density. Our choice $\lambda$ satisfies the standing assumptions in this paper. Also, both $\mu(y)$ and $\sigma^2(y)=\mu^2(y)/\lambda^2(y)$ are integrable with respect to the stationary distribution of $\Yh{}$, and $\overline\mu = 0.087$ and $\overline\sigma^2 = 0.0176$.

Notice that $\MCF_0$ and $\MCG_0$ are not trivial $\sigma$-algebra. We first generate a ``historical'' path $W_t^Y$ between $-M$ and $0$, and then evaluate each conditional expectation by the average of 500,000 paths. By the short-range dependence of $\Yh{t}$, we $M = (T/\Delta t)^{0.5}$ (cf. \cite{BaLaOpPhTa:03}), and by its roughness, we choose a fine mesh-size $\Delta t = 10^{-4}$. Then, the fast-varying factor $(\Yh{t})_{t \in [0,T]}$ \eqref{eq_Yh} is generated using the Euler scheme. We mention again that, the numerical results in Table \ref{table2} are only for illustration purpose for the reason aforementioned. %we computed the values for only a few ``omegas" denoted by $\#1, \#2,$ and $\#3$.

\begin{table}[H]
	\centering
	\caption{The value processes $\Vy_0$ \emph{vs.} $\Vyl_0$ \emph{vs.} $V_0^{\bar\pi^{(0)}, \eps}$ for the power utility case.}\label{table2}
	\begin{tabular}{|c|c|c|c|c|}\hline
		&& \#1 & \#2 & \#3   \\ \hline\hline 
		&$\Vy_0$ & 1.4645 & 1.4296 & 1.4075\\
		$\eps = 1$ & $\Vy_0 - \Vyl_0$& 0.0021 & 0.0021 & 0.0021\\
		& $\Vy_0 - V_0^{\bar\pi^{(0)},\eps}$& 0.0538 & 0.0495 & 0.0509\\ \hline 
		
		&$\Vy_0$& 1.4530 & 1.4328 & 1.4191 \\
		$\eps = 0.5$ & $\Vy_0 -\Vyl_0$& 0.0017 & 0.0018 & 0.0017\\
		& $\Vy_0 - V_0^{\bar\pi^{(0)},\eps}$& 0.0524 & 0.0475 & 0.0473\\ \hline 
				 
		&$\Vy_0$ & 1.4442 & 1.4464 & 1.4465\\
		$\eps = 0.1$ & $\Vy_0 - \Vyl_0$ & 0.0010 & 0.0010 & 0.0010\\
		& $\Vy_0 - V_0^{\bar\pi^{(0)},\eps}$ & 0.0423 & 0.0405 & 0.0400\\ \hline 
		
		&$\Vy_0$ & 1.4456& 1.4503 & 1.4522\\
		$\eps = 0.05$ & $\Vy_0 -\Vyl_0$ & 0.0006 & 0.0007 & 0.0007\\
		& $\Vy_0 - V_0^{\bar\pi^{(0)},\eps}$ & 0.0369 & 0.0366 & 0.0365\\ \hline 
		
		&$\Vy_0$ & 1.4507 & 1.4541 & 1.4563\\
		$\eps = 0.01$ & $\Vy_0 - \Vyl_0$ & 0.0002 & 0.0002 & 0.0003\\
		& $\Vy_0 - V_0^{\bar\pi^{(0)},\eps}$ & 0.0224  & 0.0249 & 0.0254\\ \hline 		 
		
	\end{tabular}
\end{table}

As expected, utility losses for both approximated strategies tend to decrease as $\eps$ goes to zero. The zeroth order strategy  $\pz$ performs well, even for not so small values of $\eps$. The relative utility loss is below $0.2\%$. Again, as expected, the ``lazy'' strategy $\bar\pi_t^{(0)}$ produces a larger utility loss, and thus underperforms $\pz$, but we still consider its performance good after observing that the $(\Vy_0 - V_0^{\bar\pi^{(0)},\eps})/\Vy_0$ is below $4\%$.

\subsection{Comparison with the Markovian case}\label{sec_comparisonMarkov}

In the Markovian case, corresponding to $H = \half$ in the modeling of $\Yh{t}$ \eqref{eq_Yh}, approximations to the value function and the optimal portfolio  have been rigorously derived in  \cite{FoSiZa:13}, and are of the form:
\begin{align}
\Vy(t, X_t) &= \frac{X_t^{1-\gamma}}{1-\gamma}e^{\frac{1-\gamma}{2\gamma}\overline{\lambda}^2(T-t)}\left[1 - \sqrt{\eps}\rho\left(\frac{1-\gamma}{\gamma}\right)^2\frac{\average{\lambda\theta'}}{2}(T-t)\right] + \MCO(\eps),\label{def_Vymarkovian}\\
\pi^\ast(t,X_t, \Yh{t}) &= \left[\frac{\lambda(\Yh{t})}{\gamma\sigma(\Yh{t})} + \sqrt{\eps} \frac{\rho(1-\gamma)}{\gamma^2\sigma(\Yh{t})}\frac{\theta'(\Yh{t})}{2}\right]X_t + \MCO(\eps), \label{def_pimarkovian}
\end{align}
where $\theta(y)$ solves the Poisson equation $\half \theta''(y) - ay\theta'(y) = \lambda^2(y) - \overline{\lambda}^2$.	These can be viewed as the limits $\lim_{\eps \to 0}\lim_{H \uparrow \half}$ of our current setup.

However, the two limits apparently do not commute for the optimal control, since there is no correction term at order $\sqrt{\eps}$ in \eqref{def_pioptimalunderfOU}. This is also the case for the problem value $\Vy$, even the first order correction in \eqref{def_Qeps} turns out to be of order $\sqrt{\eps}$. Formally, letting $H \uparrow \half$ in equation \eqref{def_Qeps}, we obtain
\begin{equation*}
\Vy(t,X_t) = \frac{X_t^{1-\gamma}}{1-\gamma}e^{\frac{1-\gamma}{2\gamma}\overline{\lambda}^2(T-t)}\left[1 + \sqrt{\eps}\rho\left(\frac{1-\gamma}{\gamma}\right)^2\overline D'(T-t)\right] + o(\sqrt\eps) 
\end{equation*}
where $\overline D'$ is the limit of $\overline D$ as $H$ approaches $\half$,
\begin{equation*}
\overline D' = \lim_{H \uparrow \half} \overline D = \int_0^\infty\left[\iint_{\RR^2} \lambda(\sigma_{ou}z)(\lambda\lambda')(\sigma_{ou}z') \tilde p_{C_Y(s)} (z,z') \right]e^{-as} \ud s 
\end{equation*}
and $\tilde p_{C_Y(s)} (z,z')$ is the bivariate normal density with mean zero and variance matrix $\begin{bmatrix}
1 &e^{-as}\\
e^{-as} & 1
\end{bmatrix}$. This is not the same constant $\average{\lambda\theta'}$ as in equation \eqref{def_Vymarkovian}.

Although the expansions of the two cases  ($H = \half$ \emph{vs.} $H \in (0, \half)$) surprisingly share the same form (a leading order term plus the first order correction at order $\sqrt{\eps}$), the coefficients are not identical. This is because our derivations in Theorem~\ref{thm_Vtpowerexpansion} and \ref{thm_piexpansion} are only valid for $H \in (0, \half)$, and the singular perturbation is ``singular'' at $H = \half$. Consequently, the order of limits $H \uparrow \half$ and $\eps \to 0$ is not interchangeable, and this leads to different expansion results. We also remark that, unlike in the case $H > \half$ where the Hurst index $H$ influences the order of first correction ($\eps^{1-H}$), here it is $\sqrt{\eps}$ no matter what value $H$ takes in $(0, \half)$.

\section{Conclusion}\label{sec_conclusion}

In this paper, we treated the portfolio optimization problem in a one-factor stochastic environment when the investor's utility is of power type. To accommodate recent empirical studies, we model this factor using a fast mean-reverting process driven by a fractional Brownian motion with Hurst index $H < \half$. Thus, its paths are rougher than the standard Bm and it has short-range dependence. Under this  setup, the value process can be represented explicitly thanks to the martingale distortion transformation (MDT), which enables us to perform an asymptotic expansion and obtain an approximation of the form: leading order term plus a correction term at order $\sqrt{\eps}$. Surprisingly, the order of the correction is not associated to the Hurst index $H$, and the fast factor $\Yh{t}$ appears in neither terms, which is a different behavior than in the cases studied in our previous work \cite{FoHu:17,FoHu2:17}. The approximation of the optimal strategy is also analyzed, and it turns out that there is no correction at the order $\sqrt{\eps}$. Nevertheless, we are still able to show that the leading order strategy, derived in Section~\ref{sec_asymppi}, is able to reproduce the problem value up to order $\sqrt{\eps}$, and therefore, it is asymptotically optimal within all admissible strategies. We remark that this is only proved in the case of power utility, as in general, the MDT is not available either with general utility or with multi-factor models, as well as the expansion of the full problem value. However, one can work within a smaller class of admissible strategies and easily extend the ``epsilon-martingale decomposition'' argument in Section~\ref{sec_asymppzpower} to obtain a weaker optimality of $\pz$. For the general utility case, this argument is very similar to our previous work \cite[Section 4]{FoHu:17,FoHu2:17}, and we did not include it here. The multi-factor case involves more techniques and will be presented in another paper in preparation \cite{Hu:18}.

\appendix
\section{Technical Lemmas}\label{app_lemmas}
 
In this section, we present several lemmas used in Section~\ref{sec_asymppower}. Note that the constants $K, K'$ in all lemmas do not depend on $\eps$ and may vary from line to line, and we denote the function $G(y)$ as
\begin{equation*}
G(y) = \half(\lambda^2(y) -\overline{\lambda}^2),
\end{equation*}
and $\lVert X\rVert_p := (\EE X^p)^{1/p}$ as the $L^p$-norm of $X$.

\begin{lem}\label{lem_moments} \quad
	\begin{enumerate}[(i)]
		\item\label{lem_psi}
		The martingale $\psi_t^\eps$ defined in \eqref{def_psi}:
		\begin{equation*}
		\psi_t^\eps = \EE\left[\int_0^T G(\Yh{s}) \ud s \Big\vert \MCG_t\right], 
		%\quad G(y) = \half(\lambda^2(y) -\overline{\lambda}^2)
		\end{equation*}
		satisfies
		\begin{equation*}
		\ud \psi_t^\eps = \vartheta_t^\eps \ud W_t^Y, \quad \vartheta_t^\eps := \int_t^T \EE\left[G'(\Yh{s})\vert\MCG_t\right]\kereps(s-t)\ud s. 
		\end{equation*}
		Moreover, for all $t \in [0,T]$ we have
		\begin{equation*}
		\EE[\lambda(\Yh{t})\vartheta_t^\eps] = \sqrt\eps \,\overline D + \widetilde D_t^\eps,
		\end{equation*}
		where $\overline D$ is the deterministic constant and $\widetilde D_t^\eps$ is of higher order than $\sqrt{\eps}$:
		\begin{equation*}
		\sup_{\eps\in[0,1]}\sup_{t \in[0,T]} \eps^{-1/2}\abs{\widetilde D_t^\eps} < \infty, \quad \text{and} \quad \forall t\in[0,T), \lim_{\eps\to 0} \eps^{-1/2}\abs{\widetilde D_t^\eps} = 0.
		\end{equation*}
		Define the random process $\kappa_t^\eps$ by
		\begin{equation}\label{def_kappa}
		\kappa_t^\eps = \int_0^t (\vartheta_s^\eps \lambda(\Yh{s}) - \sqrt{\eps}\,\overline D) \ud s.
		\end{equation}
		It is of higher order than $\sqrt{\eps}$ in $L^2$ sense uniformly in $t \in [0,T]$:
		\begin{equation}
		\lim_{\eps \to 0} \eps^{-1/2} \sup_{t\in[0,T]} \ltwonorm{\kappa_t^\eps} = 0.
		\end{equation}
		
		\item \label{lem_phi}
		The random component $\phi_t^\eps$ defined in \eqref{def_phi} has the form
		\begin{equation*}
		\phi_t^\eps = \EE\left[\int_t^T G(\Yh{s}) \ud s \Big\vert \MCG_t\right].
		\end{equation*}	
		It is a random variable with mean zero and standard deviation of order $\eps^{1-H}$ 	uniformly in $t \in [0,T]$:
		\begin{equation*}
		\sup_{\eps\in [0,1]}\sup_{t\in[0,T]} \eps^{1-H}\ltwonorm{\phi_t^\eps} < \infty.
		\end{equation*}
	
		\item\label{lem_i} The random process $I_t^\eps$ defined in \eqref{def_i}
		\begin{equation*}
		I_t^\eps = \int_0^t \left( \lambda^2(\Yh{s}) - \overline{\lambda}^2\right) \ud s,
		\end{equation*}
		satisfies 
		\begin{equation*}
		\lim_{\eps \to 0 }\sup_{t\in[0,T]} \EE[(I_t^\eps)^2] \leq K \eps^{1-H}.
		\end{equation*}

	\end{enumerate}
\end{lem}

\begin{proof}
	
	All results are slightly different versions or straightforward generalizations of lemmas in \cite[Appendix~A]{GaSo:17}, thus we omit the details here.
\end{proof}

\begin{lem}\label{lem_comparison}\quad
	\begin{enumerate}[(i)]
	\item\label{lem_Yhtilde}
	Denote by $\Yht{t} $ the  $\widetilde \PP$-stationary fractional Ornstein--Ulenbeck process, whose moving average representation is of the form
	\begin{equation*}
	\Yht{t} := \int_{-\infty}^t \kereps(t-s)\ud \widetilde W_s^Y.
	\end{equation*}
	Then, $\sup_{t\in[0,T]}\abs{\Yht{t}-\Yh{t}} \leq K\sqrt\eps$.
	\item\label{lem_varthetatilde}
	Recall the stochastic process $\vartheta_t^\eps$ defined in \eqref{def_vartheta}, and $\widetilde \vartheta_t^\eps$ defined in \eqref{def_varthettilde}:
	\begin{equation*}
		\vartheta_t^\eps := \int_t^T  \EE[G'(\Yh{s})\vert\MCG_t]\kereps(s-t)\ud s, \quad 
	\widetilde\vartheta_t^\eps := \int_t^T \widetilde \EE[G'(\Yht{s})\vert\MCG_t]\kereps(s-t)\ud s.
	\end{equation*}
	Then, $\sup_{t\in[0,T]}\abs{\widetilde\vartheta_t^\eps-\vartheta_t^\eps} \leq K\eps$.
	\end{enumerate}
\end{lem}
\begin{proof}
Both are proved by using the arguments in \cite[Lemma A.3]{FoHu2:17} and the fact $\MCK(t) \in L^1$, and we omit the details for simplicity. 
\end{proof}

\begin{lem}\label{lem_psihat}
Recall the $\widetilde \PP$-martingale $\widehat \psi_t^\eps$ defined in \eqref{def_psitilde}	
\begin{equation*}
\widehat \psi_t^\eps = \widetilde \EE\left[\int_0^T G(\Yh{s}) \ud s \vert \MCG_t\right].
\end{equation*}
Denote its martingale representation by
\begin{equation*}
\ud \widehat \psi_t^\eps = \widehat \vartheta_t^\eps \ud \widetilde W_t^Y.
\end{equation*}
Then, the process $\widehat \vartheta_t^\eps$ satisfies $\sup_{t \in [0,T]} \abs{\vartheta_t^\eps - \widehat \vartheta_t^\eps} \leq K\eps$.

\begin{proof}
We first claim that, for all $t \in [0,T]$
\begin{equation}\label{eq_DtYs}
\int_t^T \abs{\widetilde \MCD_t \Yh{s}} \ud s \leq K\sqrt{\eps},
\end{equation}
where $\widetilde \MCD$ denotes the Malliavin derivative with respect to $\widetilde W$. This is obtained by applying the derivation in \cite[Lemma 3.1]{FoHu:17} and the fact that $\int_0^T \MCK^\eps(s)\ud s$ is bounded by $K\sqrt\eps$.

Then, $\widehat \vartheta_t^\eps$ is computed as:
\begin{align*}
\widehat \vartheta_t^\eps &= \widetilde \EE\left[\int_0^T G'(\Yh{s})\widetilde \MCD_t \Yh{s} \ud s \bigg\vert \MCG_t\right]\\
& = \int_t^T \widetilde \EE\left[G'(\Yh{s})\vert\MCG_t\right]\MCK^\eps(s-t) \ud s + \widetilde \EE\left[\int_t^T G'(\Yh{s})\int_t^s \MCK^\eps(s-u)\rho \left(\frac{1-\gamma}{\gamma}\right)\lambda'(\Yh{u})\widetilde \MCD_t \Yh{u} \ud u\ud s \bigg\vert \MCG_t\right] \\
& = \widetilde \vartheta_t^\eps  + R^\eps = \vartheta_t^\eps + (\widetilde \vartheta_t^\eps -  \vartheta_t^\eps) + R^\eps,
\end{align*}
with $R^\eps$ defined as:
\begin{align*}
R^\eps := & \int_t^T \widetilde \EE\left[G''(\chi_s)(\Yh{s}-\Yht{s})\big\vert\MCG_t\right]\MCK^\eps(s-t) \ud s \\
&+  \widetilde \EE\left[\int_t^T G'(\Yh{s})\int_t^s \MCK^\eps(s-u)\rho \left(\frac{1-\gamma}{\gamma}\right)\lambda'(\Yh{u})\widetilde \MCD_t \Yh{u} \ud u\ud s \bigg\vert \MCG_t\right],
\end{align*}
and $\chi_s$ is the remainder of Taylor expansion. Given Lemma~\ref{lem_comparison}\eqref{lem_varthetatilde}, it remains to show that $R^\eps$ is bounded by $K\eps$. The first term in $R^\eps$ is guaranteed by Lemma~\ref{lem_comparison}\eqref{lem_Yhtilde} and $\MCK(t) \in L^1$, while the second term is by the boundedness of $\lambda$ and its derivatives, $\MCK(t) \in L^1$ and \eqref{eq_DtYs}.
	\end{proof}

\end{lem}

\begin{lem}\label{lem_Mj} The processes $M_t^{(j)}$, $j = 1, 2, 3$ defined in \eqref{def_M1}, \eqref{def_M2} and \eqref{def_M3} are true $\PP$-martingales.
\end{lem}
\begin{proof}
By the Burkholder--Davis--Gundy inequality, it suffices to show $\EE\left[\average{M^{(j)}}_T^{1/2}\right] < \infty$, for $j = 1, 2, 3$.

For the case $j=1$, we compute
\begin{align*}
\ud \average{M^{(1)}}_t = \lambda^2(\Yh{t}) \left(D_1\vz(t, X_t^\pz)\right)^2 \ud t \leq K\left(X_t^\pz\right)^{2-2\gamma} \ud t,
\end{align*}
by the boundedness of $\lambda$. Then since $X_t^\pz$ has $p^{th}$-moment for any $p$ uniformly in $t \in [0,T]$, we obtain   
\begin{align*}
\EE\left[\average{M^{(1)}}_T^{1/2}\right] \leq \left[\EE\int_0^T K( X_t^\pz)^{2-2\gamma} \ud t\right]^{1/2} \leq  K \sup_{t\in[0,T]}\left(\EE[(X_t^\pz)^{2-2\gamma}]\right)^{1/2} < \infty.
\end{align*}
The martingality of $M_t^{(3)}$ is obtained in the same way; while for  $M_t^{(2)}$, additional properties such as the boundedness (uniform in $\delta$) of $\vartheta_t^\eps$ and $\phi_t^\eps$ are used.

\end{proof}

\begin{lem}\label{lem_Rjgeneral} The random variable $R_{t,T}^{(j)}$, $j= 1, 2, 3$ defined in \eqref{def_R1general}-\eqref{def_R3general} 
	\begin{align*}
	&R_{t,T}^{(1)} := \int_t^T  \phi_s^\eps\left[\half(\lambda^2(\Yh{s})-\overline\lambda^2)(D_2+2D_1)D_1\vz(s,X_s^\pz)\right]\ud s,\\
	&R_{t,T}^{(2)} := \int_t^T  \rho\left(\lambda(\Yh{s})\vartheta_s^\eps-\sqrt{\eps}\,\overline D \right)D_1^2\vz(s, X_s^\pz)\ud s,\\
	&R_{t,T}^{(3)} := \int_t^T  \half\sqrt\eps\rho\overline D (\lambda^2(\Yh{s}) - \overline{\lambda}^2)(D_2 + 2D_1)\vo(s, X_s^\pz) \ud s,
	\end{align*}	
	are of order $o(\sqrt\eps)$:
\begin{equation}\label{app_Rjgeneral}
	\lim_{\eps\to 0}\eps^{-1/2}\; \EE\abs{R_{t,T}^{(j)}} = 0, \quad \forall j = 1, 2, 3.
\end{equation}

\end{lem}
\begin{proof}
The proofs here are similar to the ones in \cite[Proposition 4.1]{GaSo:17}.

For the case $j=1$, using the definition of $\vz$ and $D_k$, the property $\phi_t^\eps \sim \MCO(\eps^{1-H})$ in $L^2$ (cf. Lemma~\ref{lem_moments}\eqref{lem_phi}), and the boundedness of $\lambda$, one deduces
\begin{align*}
\EE\abs{R_{t,T}^{(1)}} \leq K \EE\abs{\int_t^T \phi_s^\eps \left(X_s^\pz\right)^{1-\gamma}\ud s} \leq K \sup_{t \in [0,T]} \ltwonorm{\phi_t^\eps} \left(\EE\int_t^T \left(X_s^\pz\right)^{2-2\gamma} \ud s\right)^{1/2}.
\end{align*}
Since $X_t^\pz$ has $p^{th}$ moment for any $p$, the above expectation is of order $\eps^{1-H}$, and thus \eqref{app_Rjgeneral} is satisfied for $j=1$.

To prove \eqref{app_Rjgeneral} with $j = 2$, we denote $t_k = t + (T-t)k/N$, $Z_s^{(2)} = D_1^2\vz(s,X_s^\pz) = K (X_s^\pz)^{1-\gamma}$ and recall $\kappa_t^\eps$ defined in \eqref{def_kappa}, thus $R_{t,T}^{(2)}$ can be written as
\begin{align*}
R_{t,T}^{(2)} &  = \sum_{k=0}^{N-1} \int_{t_k}^{t_{k+1}} Z_s^{(2)} \frac{\ud \kappa_s^\eps}{\ud s} \ud s = \sum_{k=0}^{N-1} \int_{t_k}^{t_{k+1}} Z_{t_k}^{(2)} \frac{\ud \kappa_s^\eps}{\ud s} \ud s + \sum_{k=0}^{N-1} \int_{t_k}^{t_{k+1}} (Z_s^{(2)} - Z_{t_k}^{(2)}) \frac{\ud \kappa_s^\eps}{\ud s} \ud s \\
& = \sum_{k=0}^{N-1} Z_{t_k}^{(2)} (\kappa_{t_{k+1}}^\eps - \kappa_{t_k}^\eps) + \sum_{k=0}^{N-1} \int_{t_k}^{t_{k+1}} (Z_s^{(2)} - Z_{t_k}^{(2)}) \frac{\ud \kappa_s^\eps}{\ud s} \ud s \\
&:= R_{t,T}^{(2, a)} + R_{t,T}^{(2, b)}.
\end{align*}
We then claim two properties for $Z_t^{(2)}$: (a) it has finite second moment uniformly in $\eps$ and $s \in [0,T]$, i.e. 
\begin{align}\label{eq_Z2moments}
\sup_{\eps\in[0,1]}\sup_{t \in [0,T]}\EE[(Z_t^{(2)})^2] < \infty;
\end{align}
and (b) its increments are Lipschitz in $L^2$:
\begin{equation}\label{eq_bddincrement}
\EE[(Z_u^{(2)}-Z_v^{(2)})^2] \leq K\abs{u-v}.
\end{equation}
These are ensured by the boundedness of $\lambda(\cdot)$, and the formulation \eqref{def_Xtpz} of $X_t^\pz$ in the case of power utility. In general (utility), the two properties regarding $Z_t^{(2)}$ are also satisfied under proper assumptions, see \cite[Lemma~A.6]{FoHu2:17} for further details.
 
Now we proceed to the analysis of $R_{t,T}^{(2,a)}$ and $R_{t,T}^{(2,b)}$:
\begin{align*}
\EE\abs{R_{t,T}^{(2,a)}}&\leq \sqrt 2\sum_{k=0}^{N-1} \ltwonorm{Z_{t_k}^{(2)}}[\EE(\kappa_{t_k}^\eps)^2+\EE (\kappa_{t_{k+1}}^\eps)^2]^{1/2} \leq 2N \sup_{s\in[t,T]} \ltwonorm{Z_{s}^{(2)}} \sup_{s\in[t,T]} \ltwonorm{\kappa_{s}^\eps}
\end{align*}
and is of order $o(\sqrt{\eps})$ for any fixed $N$ by Lemma~\ref{lem_moments}\eqref{lem_psi}. For the second term, using \eqref{eq_bddincrement} and the fact that $\frac{\ud \kappa_s^\eps}{\ud s}$ is of (or higher than) order $\MCO(\sqrt\eps)$, one computes
\begin{align*}
\EE\abs{R_{t,T}^{(2,b)}} %&\leq \infnorm{\lambda}\sum_{k=0}^{N-1}\int_{t_k}^{t_{k+1}} \EE\abs{(Z_s^{(1)}-Z_{t_k}^{(1)})\phi_s^\eps} \ud s  \\
&\leq K \sqrt{\eps} \sum_{k=0}^{N-1}\int_{t_k}^{t_{k+1}} \ltwonorm{Z_s^{(2)}-Z_{t_k}^{(2)}} \ud s  \leq K\sqrt\eps \sum_{k=0}^{N-1}\int_{t_k}^{t_{k+1}} (s-t_k)^{1/2} \ud s = K\sqrt\eps\frac{1}{\sqrt N},
\end{align*}
and claims
\begin{align*}
	\lim_{\eps\to 0}\eps^{-1/2}\; \EE\abs{R_{t,T}^{(1,b)}} \leq \frac{K}{\sqrt N}
\end{align*}
holds for any $N$. Thus, we have the desired result, by letting $N \to \infty$.

Proof of \eqref{app_Rjgeneral} for $j=3$  is given by repeating the same argument as in the previous case to $Z_s^{(3)} = (D_2 + 2D_1)\vo(s, X_s^\pz)$ which also satisfies the two properties since $Z_s^{(3)} = K \left(X_s^\pz\right)^{1-\gamma}$, and to $\half\sqrt\eps\rho\overline D (\lambda^2(\Yh{s}) - \overline{\lambda}^2)$ which is clearly of order $\sqrt{\eps}$.

\end{proof}

\bibliographystyle{plainnat}
\bibliography{Reference}

\begin{thebibliography}{25}
\providecommand{\natexlab}[1]{#1}
\providecommand{\url}[1]{\texttt{#1}}
\expandafter\ifx\csname urlstyle\endcsname\relax
  \providecommand{\doi}[1]{doi: #1}\else
  \providecommand{\doi}{doi: \begingroup \urlstyle{rm}\Url}\fi

\bibitem[Bardet et~al.(2003)Bardet, Lang, Oppenheim, Philippe, and
  Taqqu]{BaLaOpPhTa:03}
J.-M. Bardet, G.~Lang, G.~Oppenheim, A.~Philippe, and M.~S. Taqqu.
\newblock Generators of long-range dependent processes: a survey.
\newblock \emph{Theory and applications of long-range dependence}, pages
  579--623, 2003.

\bibitem[Biagini et~al.(2008)Biagini, Hu, {\O}ksendal, and Zhang]{BiHuOkZh:08}
F.~Biagini, Y.~Hu, B.~{\O}ksendal, and T.~Zhang.
\newblock \emph{Stochastic calculus for fractional Brownian motion and
  applications}.
\newblock Springer Science \& Business Media, 2008.

\bibitem[Cheridito et~al.(2003)Cheridito, Kawaguchi, and Maejima]{ChKaMa:03}
P.~Cheridito, H.~Kawaguchi, and M.~Maejima.
\newblock Fractional ornstein-uhlenbeck processes.
\newblock \emph{Electronic Journal of Probability}, 8\penalty0 (3):\penalty0
  1--14, 2003.

\bibitem[Coutin(2007)]{Co:07}
L.~Coutin.
\newblock An introduction to (stochastic) calculus with respect to fractional
  brownian motion.
\newblock In \emph{S{\'e}minaire de Probabilit{\'e}s XL}, pages 3--65.
  Springer, 2007.

\bibitem[Fouque and Hu(2017{\natexlab{a}})]{FoHu:16}
J.-P. Fouque and R.~Hu.
\newblock Asymptotic optimal strategy for portfolio optimization in a slowly
  varying stochastic environment.
\newblock \emph{SIAM Journal on Control and Optimization}, 5\penalty0 (3),
  2017{\natexlab{a}}.

\bibitem[Fouque and Hu(2017{\natexlab{b}})]{FoHu:17}
J.-P. Fouque and R.~Hu.
\newblock Optimal portfolio under fractional stochastic environment.
\newblock \emph{arXiv preprint arXiv:1703.06969}, 2017{\natexlab{b}}.

\bibitem[Fouque and Hu(2018)]{FoHu2:17}
J.-P. Fouque and R.~Hu.
\newblock Optimal portfolio under fast mean-reverting fractional stochastic
  environment.
\newblock \emph{SIAM Journal on Financial Mathematics}, 2018.
\newblock to appear.

\bibitem[Fouque et~al.(2000)Fouque, Papanicolaou, and Sircar]{FoPaSi:00}
J.-P. Fouque, G.~Papanicolaou, and R.~Sircar.
\newblock \emph{Derivatives in financial markets with stochastic volatility}.
\newblock Cambridge University Press, 2000.

\bibitem[Fouque et~al.(2001)Fouque, Papanicolaou, and Sircar]{FoPaSi:01}
J.-P. Fouque, G.~Papanicolaou, and R.~Sircar.
\newblock Stochastic volatility and epsilon-martingale decomposition.
\newblock In \emph{Trends in Mathematics, Birkhauser Proceedings of the
  Workshop on Mathematical Finance}, pages 152--161. Springer, 2001.

\bibitem[Fouque et~al.(2003)Fouque, Papanicolaou, Sircar, and
  S{\o}lna]{FoPaSiSo:03}
J.-P. Fouque, G.~Papanicolaou, R.~Sircar, and K.~S{\o}lna.
\newblock Short time-scale in s\&p500 volatility.
\newblock \emph{Journal of Computational Finance}, 6\penalty0 (4):\penalty0
  1--24, 2003.

\bibitem[Fouque et~al.(2011)Fouque, Papanicolaou, Sircar, and
  S{\o}lna]{FoPaSiSo:11}
J.-P. Fouque, G.~Papanicolaou, R.~Sircar, and K.~S{\o}lna.
\newblock \emph{Multiscale Stochatic Volatility for Equity, Interest-Rate and
  Credit Derivatives}.
\newblock Cambridge University Press, 2011.

\bibitem[Fouque et~al.(2015)Fouque, Sircar, and Zariphopoulou]{FoSiZa:13}
J.-P. Fouque, R.~Sircar, and T.~Zariphopoulou.
\newblock Portfolio optimization \& stochastic volatility asymptotics.
\newblock \emph{Mathematical Finance}, 2015.

\bibitem[Frei and Schweizer(2008)]{FrSc:08}
C.~Frei and M.~Schweizer.
\newblock Exponential utility indifference valuation in two brownian settings
  with stochastic correlation.
\newblock \emph{Advances in Applied Probability}, 40\penalty0 (2):\penalty0
  401--423, 2008.

\bibitem[Garnier and S{\o}lna(2016)]{GaSo:16}
J.~Garnier and K.~S{\o}lna.
\newblock Option pricing under fast-varying long-memory stochastic volatility.
\newblock \emph{arXiv preprint arXiv:1604.00105}, 2016.

\bibitem[Garnier and S{\o}lna(2017{\natexlab{a}})]{GaSo:15}
J.~Garnier and K.~S{\o}lna.
\newblock Correction to black-scholes formula due to fractional stochastic
  volatility.
\newblock \emph{SIAM Journal on Financial Mathematics}, 8\penalty0 (1),
  2017{\natexlab{a}}.

\bibitem[Garnier and S{\o}lna(2017{\natexlab{b}})]{GaSo:17}
J.~Garnier and K.~S{\o}lna.
\newblock Option pricing under fast-varying and rough stochastic volatility.
\newblock \emph{arXiv preprint arXiv:1707.00610}, 2017{\natexlab{b}}.

\bibitem[Gatheral et~al.(2014)Gatheral, Jaisson, and Rosenbaum]{roughvol}
J.~Gatheral, T.~Jaisson, and M.~Rosenbaum.
\newblock Volatility is rough.
\newblock \emph{arXiv preprint arXiv:1410.3394}, 2014.

\bibitem[Hu(2018{\natexlab{a}})]{Hu:18}
R.~Hu.
\newblock Asymptotic optimal portfolio in fast mean-reverting stochastic
  environments.
\newblock \emph{arXiv preprint arXiv:1803.07720}, 2018{\natexlab{a}}.

\bibitem[Hu(2018{\natexlab{b}})]{Hu:XX}
R.~Hu.
\newblock Asymptotic methods for portfolio optimization problem in multiscale
  stochastic environments, 2018{\natexlab{b}}.
\newblock In preparation.

\bibitem[Kaarakka and Salminen(2011)]{KaSa:11}
T.~Kaarakka and P.~Salminen.
\newblock On fractional ornstein-uhlenbeck process.
\newblock \emph{Communications on Stochastic Analysis}, 5\penalty0
  (1):\penalty0 121--133, 2011.

\bibitem[Mandelbrot and Van~Ness(1968)]{MaVa:68}
B.~B. Mandelbrot and J.~W. Van~Ness.
\newblock Fractional brownian motions, fractional noises and applications.
\newblock \emph{SIAM review}, 10\penalty0 (4):\penalty0 422--437, 1968.

\bibitem[Merton(1969)]{Me:69}
R.~C. Merton.
\newblock Lifetime portfolio selection under uncertainty: {T}he continuous-time
  case.
\newblock \emph{Review of Economics and statistics}, 51:\penalty0 247--257,
  1969.

\bibitem[Merton(1971)]{Me:71}
R.~C. Merton.
\newblock Optimum consumption and portfolio rules in a continuous-time model.
\newblock \emph{Journal of economic theory}, 3\penalty0 (4):\penalty0 373--413,
  1971.

\bibitem[Tehranchi(2004)]{Te:04}
M.~Tehranchi.
\newblock Explicit solutions of some utility maximization problems in
  incomplete markets.
\newblock \emph{Stochastic Processes and their Applications}, 114\penalty0
  (1):\penalty0 109--125, 2004.

\bibitem[Zariphopoulou(1999)]{Za:99}
T.~Zariphopoulou.
\newblock Optimal investment and consumption models with non-linear stock
  dynamics.
\newblock \emph{Mathematical Methods of Operations Research}, 50\penalty0
  (2):\penalty0 271--296, 1999.

\end{thebibliography}

\end{document}